\documentclass[english]{article}
\usepackage[T1]{fontenc}
\usepackage[latin9]{inputenc}
\usepackage{geometry}
\usepackage{color}
\usepackage{verbatim}
\usepackage{float}
\usepackage{amsmath}
\usepackage{amssymb}
\usepackage{graphicx}
\usepackage{esint}
\usepackage[authoryear]{natbib}

\makeatletter

\providecommand{\tabularnewline}{\\}
\floatstyle{ruled}
\newfloat{algorithm}{tbp}{loa}
\providecommand{\algorithmname}{Algorithm}
\floatname{algorithm}{\protect\algorithmname}

\@ifundefined{definecolor}
 {\usepackage{color}}{}
\@ifundefined{definecolor}{\usepackage{color}}{}
\usepackage{amsthm}\usepackage{multirow}

\floatstyle{ruled}
\newfloat{algorithm}{tbp}{loa}
\providecommand{\algorithmname}{Algorithm}
\floatname{algorithm}{\protect\algorithmname}

  \theoremstyle{plain}
  \newtheorem{prop}{\protect\propositionname}
  \theoremstyle{remark}
  \newtheorem{rem}{\protect\remarkname}

\usepackage{babel}
\providecommand{\propositionname}{Proposition}
  \providecommand{\remarkname}{Remark}

\usepackage{babel}

\usepackage{babel}

\makeatother

\usepackage{babel}
\begin{document}

\title{Sequential Monte Carlo Methods for High-Dimensional Inverse Problems:
A case study for the Navier-Stokes equations}

\date{}

\author{Nikolas Kantas%
\thanks{N. Kantas is with the Department of Mathematics, Imperial
College, London, SW7 2AZ, UK. e-mail: \{n.kantas@imperial.ac.uk\}.%
},\enskip{}~Alexandros Beskos~and~Ajay Jasra%
\thanks{A. Beskos and A. Jasra are with the Department of Statistics and Applied
Probability, National University of Singapore, Singapore, 117546,
SG. e-mail: \{staba,staja@nus.edu.sg\}.%
} 
}
\maketitle
\begin{abstract}
We consider the inverse problem of estimating the initial condition
of a partial differential equation, which is only observed through
noisy measurements at discrete time intervals. In particular, we focus
on the case where Eulerian measurements are obtained from the time
and space evolving vector field, whose evolution obeys the two-dimensional
Navier-Stokes equations defined on a torus. This context is particularly
relevant to the area of numerical weather forecasting and data assimilation.
We will adopt a Bayesian formulation resulting from a particular regularization
that ensures the problem is well posed. In the context of Monte Carlo
based inference, it is a challenging task to obtain samples from the
resulting high dimensional posterior on the initial condition. In
real data assimilation applications it is common for computational
methods to invoke the use of heuristics and Gaussian approximations.
As a result, the resulting inferences are biased and not well-justified
in the presence of non-linear dynamics and observations. On the other
hand, Monte Carlo methods can be used to assimilate data in a \textit{principled}
manner, but are often perceived as inefficient in this context due
to the high-dimensionality of the problem. In this work we will propose
a generic Sequential Monte Carlo (SMC) sampling approach for high
dimensional inverse problems that overcomes these difficulties. The
method builds upon ``state of the art'' Markov chain Monte Carlo
(MCMC) techniques, which are currently considered as benchmarks for
evaluating data assimilation algorithms used in practice. SMC samplers
can improve in terms of efficiency as they possess greater flexibility
and one can include steps like sequential tempering, adaptation and
parallelization with relatively low amount of extra computations.
We will illustrate this using numerical examples, where our proposed
SMC approach can achieve the same accuracy as MCMC but in a much more
efficient manner. \\
 \\
 Keywords: Bayesian inverse problems, Sequential Monte Carlo, data
assimilation, Navier-Stokes 
\end{abstract}

\section{Introduction}

We consider the \textit{\emph{inverse problem}}\textit{ }of estimating
the initial condition of a dynamical system described by a set of
partial differential equations (PDEs) based on noisy observations
of its evolution. Such problems are ubiquitous in many application
areas, such as meteorology and atmospheric or oceanic sciences, petroleum
engineering and imaging (see e.g. \citet{bennett2002inverse,evensen2009data,talagrand1987variational12,stuart2010inverse,cotter2012mcmc,kaipio2005statistical}).
In particular, we will look at applications mostly related to numerical
weather forecasting and \textit{\emph{data assimilation}}, where one
is interested in prediction of the velocity of wind or ocean currents.
There, a physical model of the the velocity vector field is used together
with observed data, in order to estimate its state at some point in
the past. This estimated velocity field is then used as an initial
condition within the PDE to generate forecasts. In this paper we focus
on the case where the model of the evolution of the vector field corresponds
to the two-dimensional (2D) \textit{\emph{Navier-Stokes}} equations
and the data \textit{\emph{consists of Eulerian observations}} of
the evolving velocity field originating from a regular grid of fixed
positions. Although the inverse problem related to the Navier-Stokes
dynamics may not be as difficult as some real applications, we believe
it can still provide a challenging problem where the potential of
our methods can be illustrated. Furthermore, the scope of our work
extends beyond this particular model and the computational methods
we will present are generic to inverse problems related with dynamical
systems.

In a more formal set-up, let $(U,\left\Vert \cdot\right\Vert _{U})$
and $(Y,\left\Vert \cdot\right\Vert _{Y})$ be given normed vector
spaces. A statistical inverse problem can be formulated as having
to find an unknown quantity $u\in U$ that generates\textcolor{red}{\emph{
}}data $y\in Y$: 
\[
y=\mathcal{G}(u)+e\,,
\]
 where $\mathcal{G}:U\rightarrow Y$ is an observation operator and
$e\in Y$ denotes a realization of the noise in the observation; see
\citet{kaipio2005statistical} for an overview. In a least squares
formulation, one may add a \textit{\emph{Tikhonov-Phillips regularization}}
term to ensure that the problem is well posed (see e.g. \citet{cotter2012mcmc,law2011evaluating,stuart2010inverse})
in which case one seeks to find the minimizer: 
\[
u^{\star}=\arg\min_{u\in U}\left(\left\Vert \Gamma^{-1/2}\left(y-\mathcal{G}(u)\right)\right\Vert _{Y}^{2}+\left\Vert \mathcal{C}^{-1/2}(u-\mathfrak{m})\right\Vert _{U}^{2}\right)\,,
\]
 where $\Gamma,\mathcal{C}$ are trace class, positive, self-adjoint
operators on $Y,U$ respectively and $\mathfrak{m}\in U$. In addition,
one may also be interested in quantifying the uncertainty related
to the estimate $u^{\star}$. This motivates following a Bayesian
inference perspective which is the one adopted in this work. Under
appropriate conditions (to be specified later; \citet{stuart2010inverse})
one can construct a posterior probability measure $\mu$ on $U$ such
that Bayes rule holds: 
\[
\frac{d\mu}{d\mu_{0}}(u)\propto l(y;u)\,,
\]
 where $\mu_{0}$ is the prior and $l(y;u)$ is the likelihood. The
prior is chosen to be a Gaussian probability measure $\mu_{0}=\mathcal{N}(\mathfrak{m},\mathcal{C})$
(i.e. a normal distribution on $U$ with mean $\mathfrak{m}\in U$
and covariance operator $\mathcal{C}$) as implied by prior knowledge
on the smoothness or regularization considerations. The likelihood,
$l(y;u)$, is a density w.r.t some reference measure on $Y$ and is
obtained from the statistical model believed to generate the data.
For example, one may use 
\[
l(y;u)=\exp(-\frac{1}{2}\left\Vert \Gamma^{-1/2}\left(y-\mathcal{G}(u)\right)\right\Vert _{Y}^{2})\,,
\]
 if a Gaussian additive noise model is adopted.

In this paper we will consider $u$ to be the unknown initial condition
of the PDE of interest. We will model the observations as a vector
of real random variables, $Y\in\mathbb{R}^{d_{y}}$, and assume $U$
is an appropriate Hilbert space. Thus, the observation operator is
closely related to the semigroup of solution operators of the PDE,
$\left\{ \Psi(\cdot,t):U\rightarrow U\right\} _{t\geq0}$, which maps
a chosen initial condition $u\in U$ to the present state $\Psi(u,t)$
at time $t\ge0$. It is straightforward both to extend Bayesian methodology
for these spaces (\citet{stuart2010inverse}) and to also ensure that
necessary differentiability and smoothness conditions are being enforced
with regards to the evolution of the vector field via the appropriate
choice of the prior measure. We will also work with periodic boundary
domains, which is a convenient choice that allows solving PDEs numerically
using a spectral Galerkin method with Fast Fourier Transforms (FFTs).
Notice that here we are confronted with an infinite-dimensional problem
as $U$ is a function space, but in practice a high-dimensional discretization
(or mesh) is used. Still it remains as an important requirement that
any computational method should be able to cope with an arbitrary
fine discretization, i.e. that it is robust to mesh refinement.

A plethora of methods have appeared in the literature to tackle such
inverse problems. Usually these adopt various heuristic approximations
when new data points are assimilated. The first successful attempt
in this direction of algorithms was based on optimization and variational
principles (\citet{dimet1986variational,sasaki1958objective}). Later,
these ideas were combined with Gaussian approximations, linearizations
and Kalman-type computations in \citet{talagrand1987variational12}
leading to the popular 3DVAR and 4DVAR. Another popular method is
the ensemble Kalman filter (enKF), which is nowadays employed by an
increasing number of weather forecasting centers; see \citet{evensen2009data}
for an overview. Although these methods have been used widely in practice,
an important weakness is that their use is not well justified for
non-linear problems and it is hard to quantify under which conditions
they are accurate (with the exception of linear Gaussian models; \citet{Legland2011}).
A different direction that overcomes this weakness is to use Monte
Carlo computations that make full use of Bayes rule to assimilate
data in a principled manner. In this paper we will refer to these
methods as `exact' given the resulting estimation error will diminish
by using more Monte Carlo samples and also in order to distinguish
with the above methods that use heuristic approximations. Recently,
exact Markov chain Monte Carlo (MCMC) methods suitable for high dimensional
inverse problems have been proposed in the literature (\citet{beskos2008mcmc,cotter2012mcmc,law2012proposals}).
This class of MCMC algorithms can be shown to be much more accurate
than the popular data assimilation algorithms mentioned earlier (see
\citet{law2011evaluating} for a thorough comparison). However, the
improvement in performance comes at a much greater computational cost,
limiting the effect of the method to providing benchmarks for evaluating
data assimilation algorithms used in practice.

In this paper, we aim to improve in terms of the efficiency of obtaining
Monte Carlo samples for Bayesian inference. We will use these accurate
MCMC methods as building blocks within Sequential Monte Carlo (SMC)
samplers (\citet{chopin2002sequential,del2006sequential}). Our work
builds upon recent advances in MCMC/SMC methodology and we will propose
a SMC sampler suitable for high-dimensional inverse problems. SMC
methods have been very successful in a wide range of relatively low-dimensional
applications (\citet{doucet:01}) and their validity has been demonstrated
by many theoretical results (see \citet{del2004feynman} for an exhaustive
review). However, they are widely considered to be impractical for
high-dimensional data assimilation applications. We believe that this
is true only when they are implemented naively or in an inappropriate
context, e.g. using for inverse problems standard particle filtering
algorithms intended for stochastic dynamics. Evidence for this claim
can be provided by recent success of SMC in high-dimensional applications
(\citet{jasra2011inference,schafer2013sequential}) as well as recent
theoretical results with emphasis on high dimensional problems (\citet{beskos2011stability,beskos2011error,schweizer2012non}).

We will propose an efficient algorithm based on the algorithm in \citet{chopin2002sequential}.
We will generate weighted samples (called particles) from a sequence
of target probability measures, $(\mu_{n})_{n=0}^{T}$, that starts
from the prior, $\mu_{0}$, and terminates at the posterior of interest
(i.e. $\mu_{T}=\mu$). This is achieved by a combination of importance
sampling, resampling and MCMC mutation steps. Several important challenges
arise when trying to use this approach for the high-dimensional problems
of interest in this paper: 
\begin{description}
\item [{Overcoming}] \textbf{weight degeneracy:} when the amount of information
in an assimilated data-point is overwhelming then the importance weights
will exhibit a very high variance. For instance at the $n$-th step
of the algorithm (when targeting $\mu_{n}$) the observations of the
velocity field about to be assimilated might exhibit a highly peaked
likelihood function relative to the previous target (and current proposal),
$\mu_{n-1}$. 
\item [{Constructing}] \textbf{effective MCMC mutation kernels:} the availability
of MCMC kernels with sufficiently mixing properties is well-known
to be critical for algorithmic efficiency of SMC (\citet{del2004feynman}).
This is extremely challenging in high dimensions since the target
distributions are typically comprised of components with widely varying
scales and complex correlation structures. 
\item [{Effective}] \textbf{design and monitoring of algorithmic performance:}
insufficient number of particles and MCMC mutation steps or inefficient
MCMC kernels might lead to a population of particles without the required
diversity to provide good estimates. Even in such an undesirable situation
standard performance indicators such as the Effective Sample Size
(ESS) can give satisfactory values and a false sense of security (this
has been noted in \citet{chopin2002sequential}). Hence the development
and use of reliable criteria to monitor performance is required and
these should be easy to compute using the particles.

\end{description}
In contrast to standard MCMC procedures, SMC samplers possess a great
amount of flexibility with design elements that can be modified according
to the particular problem at hand. Understanding some of the statistical
properties of the posterior of interest can be used to design an appropriate
(and possibly artificial) target sequence $(\mu_{n})_{n=0}^{T}$ as
well as constructing MCMC mutation steps with adequate mixing. To
overcome the difficulties mentioned above we will propose to: 
\begin{itemize}
\item employ sequential and adaptive tempering to smooth peaked likelihoods
by inserting an intermediate target sequence between $\mu_{n-1}$
and $\mu_{n}$. At each step of the algorithm, the next temperature
will be chosen automatically based on information from the particles
as proposed in \citet{jasra2011inference}. In particular for our
problem, adaptive tempering will not increase the total computational
cost too much, when more amount of tempering is performed at earlier
stages of the algorithm, which require shorter runs of the expensive
numerical solutions of the PDE. 
\item use the particles at each stage of the algorithm and adapt the MCMC
steps to the structure of the target. Regarding this point it will
be crucial to understand how to construct MCMC kernels robust to high
dimensions (as in \citet{cotter2012mcmc}). 
\item use a statistic to measure the amount of diversity (jitter) of the
particles during the MCMC mutation. We will use a particular standardized
square distance travelled by the particles during the mutation, which
to the best of our knowledge has not been used before. Good values
for this criterion might be chosen by requiring a minimum amount of
de-correlation. 
\item exploit the fact that many steps in SMC are trivially parallelizable.
This leads to high speed-ups in execution time when implemented on
appropriate hardware platforms, such as computing clusters or GPUs
(\citet{lee2010utility,murray2013rethinking}). 
\end{itemize}
Indeed, our contribution will be to combine the above points to design
a generic and efficient SMC algorithm that can be used for a variety
of inverse problems of interest to the data assimilation community.
We will demonstrate the performance of the proposed scheme numerically
on the inverse problem related to the Navier-Stokes equations, but
we expect similar performance in other problems such as the ones described
in \citet{cotter2012mcmc}.

The organization of the paper is as follows. In Section \ref{sec:Problem-formulation}
we formulate the inverse problem related to the Navier-Stokes equations
that will be used in this paper. In Section \ref{sec:Monte-Carlo-Methods}
we present the MCMC sampling procedure of \citet{beskos2008mcmc,cotter2012mcmc}
and a basic SMC sampling method. In Section \ref{sec:Extending-SMC-for}
we will extend the SMC methodology for high dimensional inverse problems.
In Section \ref{sec:Numerical-Examples} we present two numerical
examples with the inverse problem for the Navier-Stokes equations:
in the first one SMC appears to achieve the same accuracy as MCMC
at a fraction of the computational cost; in the second one it is unrealistic
to use MCMC from a computational perspective, but SMC can provide
satisfactory numerical solutions at a reasonable computational cost.
Finally, in Section \ref{sec:Discussion-and-extensions} we present
a discussion with some possible extensions and some concluding remarks.

\section{Problem formulation\label{sec:Problem-formulation}}

In this section we will give a brief description of the Navier-Stokes
equations defined on a torus, specify the observation mechanism and
present the posterior distribution of interest for the initial condition.
We will later use the problem formulated in this section as a case
study for the proposed SMC algorithm for inverse problems.

\subsection{Navier-Stokes Equations on a Torus}



We will first set up the appropriate state space and then present
the dynamics.

\subsubsection{Preliminaries}

Consider the state (or phase) space being the 2D-torus, $\mathbb{T}=[0,2\pi)\times[0,2\pi)$,
with $x\in\mathbb{T}$ being a point on the space. The initial condition
of interest is a 2D vector field $u:\mathbb{T}\rightarrow\mathbb{R}^{2}$.
We set $u=\left(u_{1}(x),u_{2}(x)\right)'$, where $u_{1},u_{2}\in L^{2}(\mathbb{T})$
and $\cdot'$ denotes vector/matrix transpose. We will define the
vorticity as: 
\[
\varpi=\varpi(x,t)=-\nabla\times u(t,x)
\]
 with the (slightly unusual) convention that clock-wise rotation leads
to positive vorticity. Let $\left|\cdot\right|$ denote the magnitude
of a vector or complex variate. For a scalar field $g:\mathbb{T}\rightarrow\mathbb{R}$
we will write $\nabla^{\perp}g=\left(-\partial_{x_{2}}g,\partial_{x_{1}}g\right)'$.
We will also consider the vector Laplacian operator: 
\[
\Delta u=\left(\partial_{x_{1}}^{2}u_{1}+\partial_{x_{2}}^{2}u_{1},\partial_{x_{1}}^{2}u_{2}+\partial_{x_{2}}^{2}u_{2}\right)'
\]
 and, for functions $\widetilde{v},v:\mathbb{T}\rightarrow\mathbb{R}^{2}$,
the operator: 
\[
(v\cdot\nabla)\,\widetilde{v}=(v_{1}\partial_{x_{1}}\widetilde{v}_{1}+v_{2}\partial_{x_{2}}\widetilde{v}_{1},v_{1}\partial_{x_{1}}\widetilde{v}_{2}+v_{2}\partial_{x_{2}}\widetilde{v}_{2})'\ .
\]
 Define the Hilbert space: 
\[
\mathbb{U}:=\left\{ \left.2\pi-\mbox{periodic trigonometric polynomials }u:\:\mathbb{T}\rightarrow\mathbb{R}^{2}\right|\:\nabla\cdot u=0\ ,\:\int_{\mathbb{T}}u(x)dx=0\right\} ,
\]
 and let $U$ be the closure of $\mathbb{U}$ with respect to the
norm in $L^{2}(\mathbb{T}){}^{2}$%
. Let also $P:(L^{2}(\mathbb{T}))^{2}\rightarrow U$ denote the Leray-Helmholtz
orthogonal projector. An appropriate orthonormal basis for $U$ is
comprised of the functions: 
\[
\psi_{k}(x)=\tfrac{k^{\perp}}{2\pi|k|}\exp\left(ik\cdot x\right)\ ,\quad k\in\mathbb{Z}^{2}\setminus\{0\}\ ,
\]
 where $k^{\perp}=(-k_{2},k_{1})^{'}$ and $i^{2}=-1$. So $k$ corresponds
to a (bivariate) frequency and the Fourier series decomposition of
an element $u\in U$ is written as: 
\[
u(x)=\sum_{k\in\mathbb{Z}^{2}/\{0\}}u_{k}\psi_{k}(x)\ ,\quad u_{k}=\langle u,\psi_{k}\rangle=\int_{\mathbb{T}}u\cdot\bar{\psi}_{k}(x)dx\ ,
\]
 for the Fourier coefficients $u_{k}$, with $\bar{\cdot}$ denoting
complex conjugate. Notice that since $u$ is real-valued we will have
$\overline{u_{k}}=-u_{-k}$.

Also we define $A=-P\Delta$ to be the Stokes operator; note that
$A$ is diagonalized in $U$ in the basis $\left\{ \psi_{k}\right\} _{k\in\mathbb{Z}^{2}\setminus\{0\}}$
with eigenvalues $\left\{ \lambda_{k}\right\} _{k\in\mathbb{Z}^{2}\setminus\{0\}}$
where $\lambda_{k}=\left|k\right|^{2}$. Fractional powers of the
Stokes operator can then be defined by the diagonalization. For any
$s\ge0$, we define $A^{s}$ as the operator with eigenvalues $\lambda_{k,s}=\left|k\right|^{2s}$
and eigenfunctions $\left\{ \psi_{k}\right\} _{k\in\mathbb{Z}^{2}\setminus\{0\}}$
and the Hilbert spaces $U^{s}\subseteq U$ as the domain of $A^{s/2}$,
that is the set of $u\in U$ such that $\sum_{k\in\mathbb{Z}^{2}/\{0\}}|k|^{2s}|u_{k}|^{2}<\infty$.

\subsubsection{The Navier Stokes equations}

The Navier-Stokes equations describe Newton's laws of motion for an
incompressible flow of fluid defined on $\mathbb{T}$. Let the flow
be initialized with $u\in U$ and consider the case where the mean
flow is zero. We will denote the time and space varying velocity field
as $v:\mathbb{T}\times[0,\infty)\rightarrow\mathbb{R}^{2}$, $v(x,t)=\left(v_{1}(x,t),v_{2}(x,t)\right)'$
and this is given as follows: 
\begin{gather*}
\partial_{t}v-\nu\Delta v+(v\cdot\nabla)\, v=f-\nabla\mathfrak{p}\ ,\\
\nabla\cdot v=0\ ,\quad\int_{\mathbb{T}}v_{\mathtt{j}}(x,\cdot)dx=0\ ,\:\mathtt{j}=1,2\ ,\\
v(x,0)=u(x)\ ,
\end{gather*}
 where $\nu>0$ is the viscosity parameter, $\mathfrak{p}:\mathbb{T}\times[0,\infty)\rightarrow\mathbb{R}$
is the pressure function, $f:\mathbb{T}\rightarrow\mathbb{R}^{2}$
an exogenous time-homogeneous forcing. We assume periodic boundary
conditions: 
\[
v_{\mathtt{j}}(\cdot,0,t)=v_{\mathtt{j}}(\cdot,2\pi,t)\ ,\quad v_{\mathtt{j}}(0,\cdot,t)=v_{\mathtt{j}}(2\pi,\cdot,t)\ ,\quad\mathtt{j}=1,2\ .
\]
 Applying the projection $P$ to $v$, we may write the equations
in the form of an ordinary differential equation (ODE) in $U$: 
\begin{equation}
\frac{dv}{dt}+\nu Av+B(v,v)=P\left(f\right)\ ,\: v\left(0\right)=u\ ,\label{eq:odeNS}
\end{equation}
 where the symmetric bi-linear form is defined as:

\[
B(v,\tilde{v})=\tfrac{1}{2}P\left((v\cdot\nabla)\,\widetilde{v}\right)+\tfrac{1}{2}P\left((\widetilde{v}\cdot\nabla)\, v\right)\ .
\]
 Intuitively, $P$ projects an arbitrary forcing $f$ into the space
of incompressible functions $U$. See \citet{robinson2001infinite,foias2001navier}
for more details. Let $\left\{ \Psi(\cdot,t):U\rightarrow U\right\} _{t\geq0}$
denote the semigroup of solution operators for the equation (\ref{eq:odeNS})
through $t$ time units. We also define the following discrete-time
semigroup, corresponding to time instances $t=n\delta$, of lag $\delta>0$
and $n=0,\ldots,T$: 
\[
G_{\delta}^{(n)}(\cdot)=\Psi(\cdot,n\delta)
\]
 with the conventions $G_{\delta}^{(0)}=I$, $G_{\delta}^{(1)}=G_{\delta}$
and $G_{\delta}^{(n)}=G_{\delta}\circ G_{\delta}^{(n-1)}$.

In practice we will use a finite but high dimensional approximation
for $G_{\delta}^{(n)}\bigl(u\bigr)$, which is obtained numerically
using a mesh for $u,v$; we will present the details of the numerical
solution of \eqref{eq:odeNS} in Section \ref{sec:Numerical-Examples}.

\subsection{A Bayesian Framework for the Initial Condition}

We will model the data as noisy measurements of the evolving velocity
field $v$ on a fixed grid of points, $x_{1},\ldots,x_{\Upsilon}$,
for $\Upsilon\ge1$. These are obtained at regular time intervals
that are $\delta$ time units apart. So the observations will be as
follows: 
\[
y_{n,\varsigma}=v\left(x_{\varsigma},n\delta\right)+\gamma\zeta_{n,\varsigma}\ ,\quad\zeta_{n,\varsigma}\stackrel{iid}{\sim}\mathcal{N}(0,1)\ ,\quad1\le\varsigma\le\Upsilon\ ,\,\,1\le n\le T\ ,
\]
 where $\gamma\ge0$ is constant and $v$ is initialized by the unknown
`true' initial vector field, $u^{\dagger}$. To simplify the expressions,
we set: 
\[
y=\big((y_{n,\varsigma})_{\varsigma=1}^{\Upsilon}\big)_{n=1}^{T}.
\]
 Performing inference with this type of data is referred to as Eulerian
data assimilation. The likelihood of the data, conditionally on the
unknown initial condition $u$, can be written as: 
\begin{equation}
l(y;u)=\frac{1}{\mathcal{Z}(y)}\prod_{n=1}^{T}\prod_{\varsigma=1}^{\Upsilon}\exp\left(-\tfrac{1}{2\gamma^{2}}\big(\, y_{n,\varsigma}-G_{\delta}^{(n)}\Bigl(u\Bigr)(x_{\varsigma})\,\big)^{2}\right)\ .\label{eq:likeli}
\end{equation}
 where $\mathcal{Z}(y)$ a normalizing constant that does not depend
on $u$. 

We will also consider the following family of priors: 
\begin{equation}
\mu_{0}=\mathcal{N}(0,\beta^{2}A{}^{-\alpha})\label{eq:prior}
\end{equation}
 with hyper-parameters $\alpha,\beta$ affecting the roughness and
magnitude of the initial vector field. This is a convenient but still
flexible enough choice of a prior; see \citet[Sections 2.3 and 4.1]{da2008stochastic}
for an introduction to Gaussian distributions on Hilbert spaces. Indeed,
when considering the Fourier domain, we have the real function constraint
for the complex conjugate coefficients ($u_{k}=-\overline{u_{-k}}$),
so we split the domain by defining: 
\[
\mathbb{Z}_{\uparrow}^{2}=\left\{ k=(k_{1},k_{2})\in\mathbb{Z}^{2}\setminus\{0\}:\: k_{1}+k_{2}>0\right\} \cup\left\{ k=(k_{1},k_{2})\in\mathbb{Z}^{2}\setminus\{0\}:\: k_{1}+k_{2}=0,\, k_{1}>0\right\} \ .
\]
 We will impose that $u_{k}=-\overline{u_{-k}}$ for $k\in\{\mathbb{Z}^{2}\setminus\{0\}\}\setminus\mathbb{Z}_{\uparrow}^{2}$.
Since the covariance operator is determined via the Stokes operator
$A$, we have the following equivalence when sampling from the prior:
\[
u\sim\mu_{0}\ \Leftrightarrow\ \mbox{Re}(u_{k}),\,\mbox{Im}(u_{k})\stackrel{iid}{\sim}\mathcal{N}(0,\tfrac{1}{2}\beta^{2}|k|^{-2\alpha})\ ,\,\, k\in\mathbb{Z}_{\uparrow}^{2}\ .
\]
 That is, $\mu_{0}$ admits the the following Karhunen-Loève expansion:
\begin{gather}
\mu_{0}=\mathcal{L}aw\big(\sum_{k\in\mathbb{Z}^{2}\setminus\{0\}}\tfrac{\beta}{\sqrt{2}}\left|k\right|^{-\alpha}\xi_{k}\,\psi_{k}\,\,\big)\ ;\label{eq:kl1}\\
\,\,\mbox{Re}(\xi_{k})\,,\,\mbox{Im}(\xi_{k})\,\stackrel{iid}{\sim}\mathcal{N}(0,1)\ ,\, k\in\mathbb{Z}_{\uparrow}^{2}\ ;\quad\xi_{k}=-\overline{\xi_{-k}}\ ,\, k\in\{\mathbb{Z}^{2}\setminus\{0\}\}\setminus\mathbb{Z}_{\uparrow}^{2}\ .\label{eq:kl2}
\end{gather}
 Thus \emph{a-priori}, the Fourier coefficients $u_{k}$ with $k\in\mathbb{Z}_{\uparrow}^{2}$
are assumed independent normally distributed, with a particular rate
of decay for their variances as $\left|k\right|$ increases.

Adopting a Bayesian inference perspective, we need to construct a
posterior probability measure $\mu$ on $U$: 
\begin{equation}
\frac{d\mu}{d\mu_{0}}(u)=\frac{1}{Z(y)}\, l(y;u)\ ,\label{eq:target}
\end{equation}
 where $Z(y)$ is the normalization constant. Due to the generality
of the state space, some care is needed here to make sure that under
the chosen prior, the mappings $G_{\delta}^{(n)}$ possess enough
regularity (i.e. they are $\mu_{0}$-measurable) and hence the change
of measure is well defined. For this reason we present below a proposition
from \citet{cotter2009bayesian} for this Eulerian data assimilation
problem: \begin{prop} \label{prop1}Assume that $f\in U$. Consider
the Gaussian measure $\mu_{0}=\mathcal{N}(0,\beta^{2}A^{-\alpha})$
on $U$ with $\beta>0$, $\alpha>1$. Then the probability measure
$\mu$ is absolutely continuous with respect to $\mu_{0}$ with Radon\textendash{}Nikodym
derivative written in (\ref{eq:likeli}). In addition, a Maximum-a-Posteriori
(MAP) estimate of the initial condition $u$ exists in the sense that
$\frac{1}{2}\Vert\beta^{-1}A{}^{\alpha/2}u\Vert^{2}-\log l(y;u)$
attains an infimum in $U^{\alpha}$. \end{prop} \begin{proof} The
first part is Theorem 3.4 of \citet{cotter2009bayesian}. The second
statement follows from the same paper by verifying Theorem 2.7 with
Lemmata 3.1-3.2 and noting that $U^{\alpha}$ is the Cameron-Martin
space of $\mu_{0}$ by Lemma 4.3; for more details see the discussion
after Theorem 3.4. \end{proof} Notice that condition $\alpha>1$
is necessary and sufficient for $A^{-\alpha}$ to be a trace-class
operator, thus also for the infinite-dimensional Gaussian measure
$\mu_{0}$ to be well-defined; see e.g.\@ \citet[Proposition 2.18]{da2008stochastic}.
In this paper we will use a zero mean function for the prior, but
this is done purely for the sake of simplicity. In fact, Proposition
\ref{prop1} is proven in \citet{cotter2009bayesian} for $\mu_{0}=\mathcal{N}(\mathfrak{m},\beta^{2}A^{-\alpha})$
with the mean functions $\mathfrak{m}\in U^{\alpha}$. 
In addition, we note that even though $G_{\delta}^{(n)}\bigl(u\bigr)$
in (\ref{eq:likeli}) will be obtained numerically in practice on
a finite dimensional mesh, Proposition~\ref{prop1} can be extended
to the posterior defined on the corresponding finite dimensional approximation
for $u$; we refer the interested reader to Theorems 2.4 and 4.3 in
\citet{cotter2010approximation}.

\section{Monte Carlo Methods for the Inverse Problem\label{sec:Monte-Carlo-Methods}}

In this section we present some Monte Carlo algorithms that can be
used for inverse problems such as the one involving the Navier-Stokes
dynamics formulated in Section \ref{sec:Problem-formulation}. We
will present first a well-established MCMC method applied in this
context and then outline a basic general-purpose SMC sampling algorithm.
We postpone the presentation of our proposed method for the next Section.
There we will combine strengths from both the algorithms in this section
to address the complex structure of the high-dimensional posteriors
of interest in this paper.

\begin{rem} We emphasize again that the algorithms presented in both
this and the next section are `exact' (as also mentioned in the Introduction).
By `exact' we mean that the estimation errors are purely due to the
finite number of Monte Carlo samples and can be made arbitrarily small.
The methods are based on solid theoretical principles and can loosely
speaking make full use of Bayes rule to assimilate the observations.
This is in contrast to heuristic methods that invoke Gaussian approximations
and Kalman type computations. Although these are commonly used in
practice for high dimensional applications, they are not theoretically
justified for non-linear dynamics (\citet{Legland2011}). \end{rem}


\subsection{A MCMC Method on the Hilbert Space}

MCMC is an iterative procedure for sampling from $\mu$, where one
simulates a long run of an ergodic time-homogeneous Markov chain $(u\left(m\right);m\geq0)$
that is $\mu$-invariant%
\footnote{We will use throughout the convention $u(m)$ to denote the $m$-th
iteration of any MCMC transition kernel.%
}. After a few iterations (burn-in) the samples of this chain can be
treated as approximate samples from $\mu$. 
There are many possible transition kernels for implementing MCMC chains,
but we will only focus on some algorithms that have been carefully
designed for the posteriors of interest in this paper and seem to
be particularly appropriate for Hilbert-space-valued measures arising
as change of measure from Gaussian laws. Other standard MCMC algorithms
(e.g.\@ Gibbs, Random-Walk Metropolis) have been successfully used
in high-dimensional applications; see e.g. \citet{MCMC:96}. Nevertheless,
in our context the posteriors possess some particularly challenging
characteristics for MCMC: 
\begin{description}
\item [{i)}] they lack a hierarchical modeling structure. When this is
present, conditional independencies of the coefficients are often
exploited using Gibbs-type samplers that attempt updates of a fraction
(or block) of Fourier coefficients (conditional to the remaining ones)
at each iteration. Here such a structure is not present and implementation
of conditional updates would require calculations over all coefficients
(or dimensions), making Gibbs-type schemes not useful in practice. 
\item [{ii)}] they are targeting an infinite-dimensional state space (in
theory). In practice a high dimensional approximation (mesh) will
be used, but we still require from a method to be valid for an arbitrary
mesh size and hence robust to mesh refinement. This is not the case
for standard Random-Walk-type algorithms that typically deteriorate
quickly with the mesh size; see \citet{hairer2011spectral} for more
details. 
\item [{iii)}] information in the observations is not spread uniformly
over the Fourier coefficients. \emph{A-posteriori} these can have
very different scaling 
ranging from the very low frequencies (where the support of the posterior
can change drastically from the prior) to the very high ones (where
the support of the posterior may closely resemble the prior). All
these different scales cannot be easily determined either analytically
or approximately making it difficult for MCMC algorithms to adjust
their proposal's step-sizes in the many different directions of the
state space. 
\end{description}
\begin{algorithm}
\vspace{0.1cm}
 
\begin{itemize}
\item Run a $\mu$-invariant Markov chain $(u\left(m\right);m\geq0)$ as
follows: 
\item Initialize $u(0)\sim\mu_{0}$. For $m\geq1$:

\begin{enumerate}
\item Propose: 
\[
\widetilde{u}=\rho\, u\left(m-1\right)+\sqrt{1-\rho^{2}}\,\mathfrak{Z}\ ,\quad\mathfrak{Z}\sim\mu_{0}\ ,
\]

\item Accept $u\left(m\right)=\widetilde{u}$ with probability: 
\begin{equation}
1\wedge\frac{l(y;\widetilde{u})}{l(y;u\left(m-1\right))}\label{eq:acc_ratio_simple}
\end{equation}
 otherwise $u\left(m\right)=u\left(m-1\right)$. 
\end{enumerate}
\end{itemize}
\caption{MCMC for High Dimensional Inverse Problems. }

\label{alg:mcmc} 
\end{algorithm}

Considerations i) and ii) have prompted the development of a family
of global-update MCMC algorithms, which are well-defined on the Hilbert
space (and thus robust upon mesh-refinement). In Algorithm \ref{alg:mcmc}
we present such an algorithm%
\footnote{The notation $\min\left\{ a,b\right\} =a\wedge b$ is being used within
the algorithm.%
} corresponding to a Metropolis accept-reject scheme that has appeared
earlier in \citet{neal1999regression} as a regression tool for Gaussian
processes and in \citet{beskos2008mcmc,cotter2012mcmc,stuart2010inverse}
in the context of high-dimensional inference. In direct relevance
to the purposes of this paper, Algorithm \ref{alg:mcmc} has been
applied in the context of data assimilation and is often used as the
`gold standard' benchmark to compare various data assimilation algorithms
as done in \citet{law2011evaluating}. One interpretation why the
method works in infinite dimensions is that Step 1 of Algorithm \ref{alg:mcmc}
provides a proposal transition kernel that preserves the Gaussian
prior $\mu_{0}$, while the posterior itself will be preserved using
the accept/reject rule in Step 2. In contrast, standard Random-Walk
Metropolis proposals (of the type $\widetilde{u}=u(m-1)+\textrm{noise}$)
would provide proposals of a distribution which is singular with respect
to the target $\mu$, and will thus be assigned zero acceptance probability.
In practice, when a finite-dimensional approximation of $u$ is used,
both the standard MCMC methods and Algorithm \ref{alg:mcmc} will
have non-zero acceptance probability, but in the limit only Algorithm
\ref{alg:mcmc} is valid. The mixing properties of the standard MCMC
transition kernels will also diminish quickly to zero upon mesh-refinement
(in addition to the acceptance probability), whereas this is not true
for Algorithm \ref{alg:mcmc} (\citet{hairer2011spectral}).

As a limitation, it has been noted often in practice that a value
of $\rho$ very close to $1$ is needed (e.g. $0.9998$ will be used
later on) to achieve a reasonable average acceptance probability (say
$0.2-0.3$). This is due to the fact that the algorithm is optimally
tuned to the prior Gaussian measure $\mu_{0}$, whereas the posterior
resembles closely $\mu_{0}$ only at the Fourier coefficients of very
high frequencies. This leads to small exploration steps in the proposal
and relatively slow mixing of the MCMC chain, which means that one
needs to run the chain for an excessive number of iterations (of the
order $10^{6}$) to get a set of samples with reasonable quality.
In addition, each iteration requires running a PDE solver until time
$T$ to compute $l(y;u)$ in Step 2, so the approach is very computationally
expensive. To sum up, although Algorithm \ref{alg:mcmc} has provided
satisfying results in many applications (\citet{cotter2012mcmc}),
there is still a great need for further algorithmic development towards
improving the efficiency of Monte Carlo sampling.

\begin{rem}

More elaborate MCMC proposals have appeared in \citet{law2012proposals}
that can achieve better performance regarding consideration iii) above.
There the proposals are based on some approximations of the posterior,
$\mu$, following the ideas in \citet{martin2012stochastic}. The
method in \citet{law2012proposals} has appeared in parallel to the
work presented here and contains interesting ideas that are definitely
relevant to the material in this paper. Nevertheless, as it is currently
under closer investigation it will not be considered here or later
on in our comparisons in Section \ref{sec:Numerical-Examples}.

\end{rem}

\subsection{A generic SMC Approach\label{sub:A-generic-SMC}}

We proceed by a short presentation of SMC and refer the reader to
\citet{chopin2002sequential,del2006sequential} for a more thorough
treatment. Instead of a single posterior over all observations, consider
now a sequence of probability measures $(\mu_{n})_{n=0}^{T}$ defined
on $U$ such that $\mu_{T}=\mu$ and $\mu_{0}$ is a prior as in (\ref{eq:prior}).
For example, one may consider the natural ordering of the observation
times to construct such a sequence. Indeed, consider the likelihood
of the block of observations at the $p$-th epoch: 
\[
l_{p}(y_{p};u):=\frac{1}{\mathcal{Z}_{p}(y_{p})}\prod_{\varsigma=1}^{\Upsilon}\exp\left(-\tfrac{1}{2\gamma^{2}}\big(\, y_{p,\varsigma}-G_{\delta}^{(p)}\biggl(u\biggr)(x_{\varsigma})\,\big)^{2}\right)\ .
\]
 Note that as $p$ increases so does the computational effort required
to compute $l_{p}$ due to using a numerical PDE solution to evaluate
$G_{\delta}^{(p)}\left(u\right)$. Given that the observation noise
is independent between different epochs, we define a sequence of posteriors
~$(\mu_{n})_{n=0}^{T}$ as follows: 
\begin{equation}
\frac{d\mu_{n}}{d\mu_{0}}\left(u\right)=\frac{1}{Z_{n}}\prod_{p=1}^{n}l_{p}(y_{p};u)\ ,\quad0\le n\le T\ .\label{eq:smc_sequence}
\end{equation}
 This forms a bridging sequence of distributions between the prior
and the posterior, which also admits a Karhunen-Loève expansion: 
\begin{equation}
\mu_{n}=\mathcal{L}aw\big(\sum_{k\in\mathbb{Z}^{2}\setminus\{0\}}\tfrac{\beta}{\sqrt{2}}\left|k\right|^{-\alpha}\xi_{k,n}\,\psi_{k}\,\,\big)\ ,k\in\mathbb{Z}_{\uparrow}^{2}\ ;\quad\xi_{k,n}=-\bar{\xi}_{-k,n}\ ,\, k\in\{\mathbb{Z}^{2}\setminus\{0\}\}\setminus\mathbb{Z}_{\uparrow}^{2}\ ,\label{eq:kl1-1}
\end{equation}
 where compared to \eqref{eq:kl1}-\eqref{eq:kl2}, $\left\{ \xi_{k,n}\right\} _{k\in\mathbb{Z}_{\uparrow}^{2}}$
are now correlated random variables from some unknown distribution.
Note the particular choice of $(\mu_{n})_{n=0}^{T}$ in \eqref{eq:smc_sequence}
is a natural choice for this problem. In fact, there are other alternatives
involving artificial sequences and introduction of auxiliary variables
(and the extension in the next section are an example for this; see
\citet{del2006sequential} for some more examples). The SMC algorithm
will target sequentially each intermediate $\mu_{n}$, which will
be approximated by a weighted swarm of $N\ge1$ particles (or samples).
This is achieved by a sequence of selection and mutation steps (see
\citet[Chapter 5]{del2004feynman}): 
\begin{description}
\item [{Selection}] \textbf{step:} At the $n$-th iteration say we have
available $N$ equally weighted samples of $\mu_{n-1}$, denoted $\{u_{n-1}^{j}\}_{j=1}^{N}$.
These will be used as importance proposals for $\mu_{n}$ and are
assigned the incremental (normalized) weights: 
\[
W_{n}^{j}\propto\frac{d\mu_{n}}{d\mu_{n-1}}(u_{n-1}^{j})=l_{n}(y_{n};u_{n-1}^{j}),\quad\sum_{j=1}^{N}W_{n}^{j}=1,\quad1\le j\le N\ .
\]
 The weighting step is succeeded by a resampling step so as to discard
samples with low weights. The particles are resampled probabilistically
with replacement according to their weights $W_{n}^{j}$. 
\item [{Mutation}] \textbf{step:} Carrying out only selection steps will
eventually lead to degeneracy in the diversity of the particle population.
During each successive resampling step, only few parent particles
will survive and copy themselves. Thus, some `jittering' of the population
is essential to improve the diversity. These jittering steps should
maintain the statistical properties of $\mu_{n}$ and are hence chosen
to be a small number of $\mu_{n}$-invariant MCMC iterations. For
example, one could consider using a few times Steps 1-2 of Algorithm
\ref{alg:mcmc} but with the complete likelihood $l$ replaced with
$\prod_{s=1}^{n}l_{s}$ (although we will discuss later why this is
not recommended). 
\end{description}
In Algorithm \ref{alg:basic_smc} we present the general purpose SMC
algorithm that that has appeared in \citet{chopin2002sequential}.
For the resampling step, we have used $\mathcal{R}$ to denote the
distribution of the indices of the parent particles. For instance,
one may copy offsprings according to successful counts based on the
multinomial distribution of the normalized weights (this is the approach
we follow in this paper). Recall also that $u_{n}^{j}$ denotes the
$j$-th particle approximating $\mu_{n}$ and in this paper this will
be thought as a concatenated vector of the real and imaginary parts
of the Fourier coefficients in $\mathbb{Z}_{\uparrow}^{2}$ (or its
finite truncation). Upon completion of Step 2 one obtains particle
approximations for $\mu_{n}$: 
\[
\mu_{n}^{N}=\sum_{j=1}^{N}W_{n}^{j}\delta_{u_{n}^{j}}\ ,
\]
 where $\delta_{u}$ denotes the Dirac point measure at $u\in U$.
The convergence and accuracy of these Monte-Carlo approximations have
been established under relatively weak assumptions and in various
contexts; see \citet{del2004feynman} for several convergence results.
Note that most steps in the algorithm allow for a trivial parallel
implementation and hence very fast execution times; see \citet{lee2010utility,murray2013rethinking}
for more details. In addition, the resampling step is typically only
performed when an appropriate statistic (commonly the effective sample
size (ESS)) will indicate its necessity; e.g. when ESS will drops
below a prescribed threshold: 
\[
\mbox{ESS}_{n}=(\sum_{j=1}^{N}\left(W_{n}^{j}\right)^{2})^{-1}<N_{thresh}.
\]
 When not resampling, particles keep their different weights $W_{n}^{j}$
(and are not all set to $1/N$), which are then multiplied with the
next incremental weights. In \citet{chopin2002sequential} the author
uses the particle population to design $\mathcal{K}_{n}$ (that) either
as a standard random walk or as an independent Metropolis-Hastings
sampler based on particle approximation $\mu_{n}^{N}$.

\begin{algorithm}
\vspace{0.1cm}
 
\begin{itemize}
\item At $n=0$, for $j=1,\ldots N$ sample $u_{0}^{j}\sim\mu_{0}$. 
\item Repeat for $n=1,\ldots,T$:

\begin{enumerate}
\item Selection:

\begin{enumerate}
\item Importance Sampling: weight particles \, $W_{n}^{j}\propto W_{n-1}^{j}\frac{d\mu_{n}}{d\mu_{n-1}}(u_{n-1}^{j})$\ ,
$\,\,\sum_{j=1}^{N}W_{n}^{j}=1$\ . 
\item Resample (if required):

\begin{enumerate}
\item Sample offsprings $\left(p_{n}^{1},\ldots,p_{n}^{N}\right)\sim\mathcal{R}(W_{n}^{1},\ldots,W_{n}^{N})$. 
\item Set $\breve{u}_{n}^{j}=u_{n-1}^{p_{n}^{j}}$ and $W_{n}^{j}=\frac{1}{N}$,
$1\le j\le N$. 
\end{enumerate}
\end{enumerate}
\item $\mu_{n}$-invariant mutation: update $u_{n}^{j}\sim\mathcal{K}_{n}(\breve{u}_{n}^{j},\cdot)$,
$1\le j\le N$, where $\mu_{n}\mathcal{K}_{n}=\mu_{n}$. 
\end{enumerate}
\end{itemize}
\caption{Basic Sequential Monte Carlo}

\label{alg:basic_smc} 
\end{algorithm}

As Algorithm \ref{alg:basic_smc} is based on sequential importance
sampling its success relies on: 
\begin{itemize}
\item $\mu_{n-1}$ resembling closely $\mu_{n}$ in order to avoid the weights
degenerating. 
\item $\mathcal{K}_{n}$ providing sufficient jitter to the particles in
order to counter the lost diversity due to resampling. For instance,
the number of MCMC iterations used needs to be enough for particles
to spread around the support of $\mu_{n}$. 
\end{itemize}
These issues become much more pronounced in high dimensions. We will
extend the SMC methodology to deal with these issues in the following
section.

\section{Extending SMC for High Dimensional Inverse Problems\label{sec:Extending-SMC-for}}

One advantage of SMC is its inherent flexibility due to all different
design elements such as the sequence $(\mu_{n})_{n=1}^{T}$ or the
kernels $\mathcal{K}_{n}$. In high dimensional applications such
as data assimilation a user needs to design these carefully to obtain
good performance. In addition, monitoring the performance also includes
some challenges itself. We will deal with these issues in this section.

Firstly, recall the equivalence between representing the initial vector
field $u$ by its Fourier coefficients and vice versa: 
\[
u\leftrightarrow\{u_{k}\}_{k\in\mathbb{Z}_{\uparrow}^{2}}\ ;\quad u_{k}=\langle u,\psi_{k}\rangle,\:\bar{u}_{k}=-u_{k}.
\]
 In this section $u$ will be treated again as the concatenated vector
of the real and imaginary part of its Fourier coefficients. In theory,
this vector is infinite-dimensional, but in practice it will be finite
(but high) dimensional due to the truncation in the Fourier space
used in the numerical PDE solver. We will sometimes refer to the size
of the implied mesh, $d_{u}$, informally as the dimensionality of
$u$.

SMC as in Algorithm \ref{alg:basic_smc} will need to satisfy some
requirements to be effective. As already mentioned at the end of Section
\ref{sec:Monte-Carlo-Methods}, broadly speaking these are: 
\begin{itemize}
\item each selection step should not deplete the particle population by
excessively favoring one or relatively few particles. This deficiency
can arise in the context of importance sampling when likelihood densities
are overly peaked, so that only a few particles give non-trivial likelihood
values to the corresponding observation. 
\item each mutation step should sufficiently jitter the particles, so that
the population will span most of the support of the target measure.
Ideally, the mixing properties of the MCMC kernels assigned to do
this should not degrade with increasing $n$. 
\end{itemize}
We will address the first point by altering the sequence of SMC targets
and the second point by proposing improved MCMC kernels compared to
simple modifications of Algorithm \ref{alg:mcmc}.

First we need to ensure that the importance sampling weights (in Step
1 of Algorithm \ref{alg:basic_smc}) are `stable' in the sense that
they exhibit low variance. For high dimensional inverse problems it
is expected that this is not the case when the sequence $(\mu_{n})_{n=1}^{T}$
is defined as in (\ref{eq:smc_sequence}). We will modify the sequence
of target distributions $(\mu_{n})_{n=1}^{T}$ so that it evolves
from the prior $\mu_{0}$ to the posterior $\mu$ more smoothly or
in a more `stable' manner. One way to achieve this is by bridging
the two successive targets $\mu_{n-1}$ and $\mu_{n}$ via intermediate
tempering steps as in \citet{neal2001annealed}. So one can introduce
a (possibly random) number, say $q_{n}$, of artificial intermediate
targets between $\mu_{n-1}$ and $\mu_{n}$: 
\begin{equation}
\mu_{n,r}=\mu_{n-1}\big(\tfrac{d\mu_{n}}{d\mu_{n-1}}\big)^{\phi_{n,r}}\ ,\label{eq:temper_target}
\end{equation}
 where 
\begin{equation}
0=\phi_{n,0}<\phi_{n,1}<\cdots<\phi_{n,q_{n}}=1\ ,\label{eq:number_freq}
\end{equation}
 are a sequence of user-specified temperatures. The accuracy of SMC
when using such tempering schemes have been the topic of study in
\citet{beskos2011error,beskos2011stability,giraud2012non,schweizer2012non,whiteley2012sequential}.
From these works the most relevant to the present high-dimensional
setup are \citet{beskos2011stability,beskos2011error}. There in a
slightly simpler setup the authors demonstrated that it is possible
to achieve the required weight stability with a reasonable computational
cost, when the sequence of targets varies slowly and the MCMC mutation
steps mix well for every target in the bridging sequence.

For the SMC sequence implied jointly by (\ref{eq:smc_sequence}) and
\eqref{eq:temper_target}, in Section \ref{sub:tempering} we will
present an adaptive implementation for choosing the next temperature
on-the-fly (\citet{jasra2011inference}) and in Section \ref{sub:Improving-the-mixing}
propose improved MCMC mutation kernels that use particle approximations
for each $\mu_{n,r}$.

\subsection{Stabilizing the Weights with Adaptive Tempering}

\label{sub:tempering}

A particularly useful feature of using tempering within SMC is that
one does not need to choose for every bridging sequence $q_{n}$ and
$\phi_{n,0},\ldots,\phi_{n,q_{n}}$ before running the algorithm.
In fact these can be decided on-the-fly using the particle population
as it was originally proposed in \citet{jasra2011inference}. Suppose
at the moment the SMC algorithm is about to proceed to iteration $n,r$
(the $r$-th tempering step between $\mu_{n-1}$ and $\mu_{n}$).
The MCMC mutation step for temperature $\phi_{r-1,n}$ has just completed
and let $\{u_{n,r-1}^{j}\}_{j=1}^{N}$ be equally weighted particles%
\footnote{Here $u_{n,r-1}^{j}$ denotes the concatenated real vector of real
and imaginary Fourier coefficients in $\mathbb{Z}_{\uparrow}^{2}$
for the $j$-th particle targeting $\mu_{n,r-1}$. Often in the discussion
we interpret $\mu_{n,r}$ as a probability measure on a similar real
vector $u_{n,r}$.%
} approximating $\mu_{n,r-1}$ as defined in (\ref{eq:temper_target}).
The next step is to use $\{u_{n,r-1}^{j}\}_{j=1}^{N}$ as importance
proposals for $\mu_{n,r}$. 
The incremental weights are equal to $W_{n,r}^{j}\propto\frac{d\mu_{n,r}}{d\mu_{n,r-1}}(u_{n,r-1}^{j})$,
so if $\phi_{n,r}$ has been specified, they can be also written as:
\begin{equation}
W_{n,r}^{j}=\frac{l_{n}(y_{n};u_{n,r-1}^{j})^{\phi_{n,r}-\phi_{n,r-1}}}{\sum_{s=1}^{N}l_{n}(y_{n};u_{n,r-1}^{s})^{\phi_{n,r}-\phi_{n,r-1}}}\ .\label{eq:weights}
\end{equation}
 Now from the expression in \eqref{eq:weights} it follows that one
can choice of $\phi_{n,r}$ by imposing a minimum quality for the
particle population after the weighting, e.g. a minimum value for
the ESS. Therefore we can use the particles $\{u_{n,r-1}^{j}\}_{j=1}^{N}$
to specify $\phi{}_{n,r}$ as the solution of the equation: 
\begin{equation}
\mbox{ESS}_{n,r}(\phi_{n,r})=\big(\sum_{j=1}^{N}\left(W_{n,r}^{j}\right)^{2}\big)^{-1}\approx N_{thresh}\ .\label{eq:ess}
\end{equation}
 If $\mbox{ESS}_{n,r}\left(1\right)>N_{thresh}$ one should set $\phi_{n,r}=1$
and proceed to the next tempering sequence leading to $\mu_{n+1}$.
Solving the above equation for $\phi_{n,r}$ can be easily implemented
using an iterative bisection on $(\phi_{n,r-1},1]$ for the scalar
function $\mbox{ESS}_{n,r}(\phi)$. There is now only one user-specified
parameter to be tuned, namely $N_{thresh}$. This adaptive tempering
approach of \citet{jasra2011inference} has been also used successfully
in \citet{schafer2013sequential,zhou2013towards}. In \citet{zhou2013towards}
one may also find an alternative choice for the quality criterion
instead of the ESS. Finally, assuming sufficient mixing for the MCMC
mutation steps, the accuracy of adaptive tempering has been studied
in \citet{giraud2012non}.

\begin{rem} \label{remark:temper} The exposition uses intermediate
tempering between the natural target sequence $(\mu_{n})_{n=0}^{T}$
defined in \eqref{eq:smc_sequence}. If one is not interested in the
intermediate posterior using part of the data then it is possible
to attempt tempering directly on $\mu$. Similarly, if each $y_{n}$
is a high dimensional vector, then one could define intermediate targets
between each or some elements in $y_{n}$ and temper between these
targets. In the latter case caution must be taken to avoid unnecessary
intermediate tempering steps due to outliers in the observations.
In any of these cases we remark the presented methodology still applies.
\end{rem}


\subsection{Improving the Mixing of MCMC Steps with Adaptive Scaling}

\label{sub:Improving-the-mixing}

We proceed by considering the design of the MCMC mutation steps to
be used between tempering steps. We will design a random-walk-type
method, tuned to the structure of the target distributions, by combining
two ingredients: 
\begin{description}
\item [{(a)}] we will use current information from the particles to adapt
the proposal on-the-fly to the target distribution. 
\item [{(b)}] we will distinguish between high and low frequencies for
the MCMC formulation. This is a consideration specific to inverse
problems related to dissipative PDEs, where the data often contains
more information about the lower frequencies. 
\end{description}
We will look for a moment at the MCMC mutation kernel of Algorithm
\ref{alg:basic_smc}. Recall the correspondence between an element
in $U$ and its Fourier coefficients. To remove the effect of different
scaling for each Fourier coefficient (due to the different variances
in prior) we will consider the bijection 
\[
u\leftrightarrow\{\xi_{k}\}_{k\in\mathbb{Z}_{\uparrow}^{2}}
\]
 as implied by (\ref{eq:kl1})-(\ref{eq:kl2}) for $\mu_{0}$ or \eqref{eq:kl1-1}
for $\mu_{n}$. As a result \emph{a-priori,} $\mbox{Re}(\xi_{k})$,
$\mbox{Im}(\xi_{k})$ for all $k\in\mathbb{Z}_{\uparrow}^{2}$ are
i.i.d. samples from $\mathcal{N}(0,1)$. The proposal in Step 1 of
Algorithm \ref{alg:mcmc} written in terms of $\xi_{k}$ is: 
\begin{equation}
\widetilde{\xi}_{k}=\rho\,\xi_{k}+\sqrt{1-\rho^{2}}\,\mathfrak{Z}_{k}\ ,\quad\mbox{Re}(\mathfrak{Z}_{k}),\,\mbox{Im}(\mathfrak{Z}_{k})\stackrel{iid}{\sim}\mathcal{N}(0,1)\ ,\quad k\in\mathbb{Z}_{\uparrow}^{2}\ .\label{eq:rw_prior}
\end{equation}
 This would be an excellent proposal when the target is the prior
$\mu_{0}$ (or close to it). When such a proposal is used within MCMC
transition kernels for Step 2 of Algorithm \ref{alg:basic_smc}, then
the mixing of the resulting mutation kernels will rapidly deteriorate
as we move along the sequence of bridging distributions between $\mu_{0}$
and $\mu$. The assimilated information from the observations will
change the posterior densities for each $\xi_{k}$ relative to the
prior. In particular, often the data will contain a lot of information
for the Fourier coefficients located at low frequencies, thereby shrinking
their posterior variance. Thus, at these low frequencies the update
in (\ref{eq:rw_prior}) will require a choice of $\rho$ very close
to 1 for the proposal $\widetilde{\xi_{k}}$ to have a non-negligible
chance to remain within the domain of the posterior and hence deliver
non-vanishing acceptance probabilities. At the same time such small
steps will penalize the mixing of the rest of the Fourier coefficients
with relatively large posterior variances. This is a well known issue
often seen in practice (\citet{law2012proposals}) and somehow the
scaling of the random walk exploration for each frequency needs to
be adjusted to the shape of the posterior it is targeting.

We will adapt the proposal to the different posterior scalings in
the coefficients using the particles. Assume that the algorithm is
currently at iteration $n,r$, where the importance sampling step
with proposals from $\mu_{n-1,r}$ in (\ref{eq:temper_target}) has
been completed and we have the weighted particle set $\{u_{n,r-1}^{j},W_{n,r}^{j}\}_{j=1}^{N}$
approximating $\mu_{n,r}$. We will construct the MCMC mutation kernel
$\mathcal{K}_{n,r}$ (so that $\mu_{n,r}\mathcal{K}_{n,r}=\mu_{n,r}$)
as follows. With a slight abuse of notation, we denote by $u_{k,n,r-1}^{j}$
the bivariate real vector comprised of the real and imaginary part
of the $k$-th Fourier coefficient of $u_{n,r-1}^{j}$, where $1\le j\le N$
and $k\in\mathbb{Z}_{\uparrow}^{2}$. We estimate the marginal mean
and covariance of the $k-$th Fourier coefficient under the current
target in the sequence $\mu_{n,r}$ as follows: 
\begin{equation}
\mathfrak{m}_{k,n,r}^{N}=\sum_{j=1}^{N}W_{k,n,r}^{j}u_{k,n,r-1}^{j}\ ,\quad\Sigma_{k,n,r}^{N}=\sum_{j=1}^{N}W_{k,n,r}^{j}(u_{k,n,r-1}^{j}-\mathfrak{m}_{k,n,r}^{N})(u_{k,n,r-1}^{j}-\mathfrak{m}_{k,n,r}^{N})'\ .\label{eq:mean-covar}
\end{equation}
 The estimated moments $\mathfrak{m}_{k,n,r}^{N}$, $\Sigma_{k,n,r}^{N}$
(and the corresponding Gaussian approximation of the posterior) will
be used to provide the scaling of the step size of the random walk
for each $k$. Let $\{\breve{u}_{n,r}^{j}\}_{j=1}^{N}$ be the collection
of particles obtained after resampling $\{u_{n,r-1}^{j}\}_{j=1}^{N}$
with replacement according to the weights $W_{n,r}^{j}$. In the MCMC
mutation step we can use the following proposal instead of (\ref{eq:rw_prior}):
\begin{equation}
\widetilde{u}_{k,n,r}=\mathfrak{m}_{k,n,r}^{N}+\rho\,(\breve{u}_{k,n,r}^{j}-\mathfrak{m}_{k,n,r}^{N})+\sqrt{1-\rho^{2}}\,\mathcal{N}(0,\Sigma_{k,n,r}^{N})\ .\label{eq:proposal_rw_post}
\end{equation}
 Notice that in \eqref{eq:proposal_rw_post} we propose to move the
real and imaginary parts of the Fourier coefficients separately for
each frequency $k\in\mathbb{Z}_{\uparrow}^{2}$. This requires computing
only the mean and covariances of the $k$-th marginal of $\mu_{n,r}$.
Other options are available, involving jointly estimating higher-dimensional
covariance matrices $\Sigma_{n,r}^{N}$ (for the joint vector $u_{n,r}$
involving every $k$). However, caution is needed because of the Monte
Carlo variance in estimating high dimensional covariances, when the
number of particles is moderate (as it will be the case for the computationally
expensive numerical examples in this paper). A pragmatic option in
this case is to use only the diagonal elements of the joint covariance
estimator or possibly include some off-diagonal terms where (partial)
correlations are high.

\citet{chopin2002sequential} has applied such adaptive ideas advocating
both the use of an independence sampler as well as a standard random
walk (e.g. $\widetilde{u}_{n,r}=\breve{u}_{n,r}^{j}+\varrho\mathcal{N}(0,\Sigma_{n,r}^{N})$).
There more emphasis is placed towards the independence sampler, as
it was possible to use asymptotic statistical theory to suggest that
the posterior converges with $n$ to a multivariate normal. Thus,
accurate estimation of the covariance matrix would eventually result
to high acceptance probabilities. This approach is better suited for
low dimensions (and a high number of particles), but in high dimensions
it makes sense to opt for a local-move proposal like \eqref{eq:proposal_rw_post}.
In high dimensions it is extremely hard to capture very accurately
the structure of the target $\mu_{n,r}$ as in \citet{chopin2002sequential}
simply by using estimates of very high-dimensional mean vectors and
covariance matrices based on moderate number of Monte Carlo samples.
This intuition was confirmed also in \citet{schafer2013sequential},
where a different local-move approach was proposed for high dimensional
binary spaces using random walks from appropriate parametric families
of distributions.


The second ingredient of the proposed MCMC mutation kernel involves
distinguishing between low and high frequencies. In particular, we
will use the proposal of (\ref{eq:proposal_rw_post}) for a window
of the Fourier coefficients with relatively low frequencies and the
standard proposal of (\ref{eq:rw_prior}) that uses the prior for
the higher frequencies. 
We have found empirically that this thrifty hybrid approach gives
a better balance between adaptation and variability caused by Monte
Carlo error in estimating empirical covariances. We will use the proposal
of (\ref{eq:proposal_rw_post}) for coefficients with frequencies
in the rectangular window defined as: 
\begin{equation}
\mathbf{K}=\left\{ k\in\mathbb{Z}^{2}\setminus\{0\}:\: k_{1}\vee k_{2}\leq K\right\} ,\label{eq:K}
\end{equation}
where $k_{1}\vee k_{2}=\max\{k_{1},k_{2}\}$. It would certainly be
possible to use an alternative definition for $\mathbf{K}$. For instance,
we could include all frequencies such that $|k|\le K$, but this is
not expected to make a considerable difference in terms of computational
efficiency. The idea here is that for high enough frequencies the
marginal posterior of the Fourier coefficients should be very similar
to the prior, i.e.\@ the observations are not very informative for
these frequencies. Hence, in high frequencies one might as well use
the standard proposal of (\ref{eq:rw_prior}) and still attain good
enough mixing. 

\begin{algorithm}[t]
\vspace{0.1cm}

\begin{enumerate}
\item Initialize $u_{n,r}(0)=\mathfrak{u}$ (when in Step $3(f)$ of Algorithm
\ref{alg:smc_adapted} set $u_{n,r}(0)=\breve{u}_{n,r}^{j}$). Let
$\mathfrak{m}_{k,n,r}^{N}$, $\Sigma_{k,n,r}^{N}$ be known approximations
for all {\small $\mathbf{K}\cap\mathbb{Z}_{\uparrow}^{2}$}. Choose
$\rho_{L},\rho_{H}\in(0,1)$. 
\item For $m=0,\ldots,M-1$:

\begin{enumerate}
\item For $k\in\mathbf{K}\cap\mathbb{Z}_{\uparrow}^{2}$, propose the update:
\[
\widetilde{u}_{k}=\mathfrak{m}_{k,n,r}^{N}+\rho_{L}(u_{k}(m)-\mathfrak{m}_{k,n,r}^{N})+\sqrt{1-\rho_{L}^{2}}\,\,\mathcal{N}(0,\Sigma_{k,n,r}^{N})\ ;
\]
 for $k\in\mathbf{K}^{c}\cap\mathbb{Z}_{\uparrow}^{2}$ propose the
update: 
\[
\widetilde{u}_{k}=\rho_{H}\, u_{k}(m)+\sqrt{1-\rho_{H}^{2}}\,\,\mathcal{N}(0,\tfrac{1}{2}\beta^{2}|k|^{-2\alpha}\, I)\ .
\]
 
\item Compute forward dynamics $G_{\delta}^{(n)}(\widetilde{u})$ and likelihood
functions $l_{1}(y_{1};\widetilde{u}),\ldots,l_{n}(y_{n};\widetilde{u})$. 
\item With probability: 
\begin{equation}
1\wedge\frac{l_{n,r}(\widetilde{u})}{l_{n,r}(u(m))}\frac{\mu_{0}(\widetilde{u}_{L})}{\mu_{0}(u_{L}(m))}\frac{\mathcal{Q}(\widetilde{u}_{L},u_{L}(m))}{\mathcal{Q}(u_{L}(m),\widetilde{u}_{L})}\label{eq:acc_ratio}
\end{equation}
 accept the proposal and set $u(m+1)=\widetilde{u}$\ ; otherwise
set $u(m+1)=u(m)$. We use: 
\end{enumerate}

\begin{align}
l_{n,r}(u) & =l_{n}(y_{n};u)^{\phi_{n,r}}\prod_{s=1}^{n-1}l_{s}(y_{s};u)\ ,\\
\mu_{0}(u_{L}) & =\exp\Big\{-\sum_{k\in\mathbf{K}\cap\mathbb{Z}_{\uparrow}^{2}}\beta^{-2}|k|^{2\alpha}\left|u_{k}\right|^{2}\Big\}\ ,\label{eq:mu_0_acc}\\
\mathcal{Q}(u_{L},\tilde{u}_{L}) & =\exp\Big\{\,-\tfrac{1}{2(1-\rho_{L}^{2})}\sum_{k\in\mathbf{K}\cap\mathbb{Z}_{\uparrow}^{2}}\left(\tilde{u}_{k}-\mathfrak{m}_{k,n,r}^{N}-\rho_{L}(u_{k}-\mathfrak{m}_{k,n,r}^{N})\right)'\left(\Sigma_{k,n,r}^{N}\right){}^{-1}\left(\tilde{u}_{k}-\mathfrak{m}_{k,n,r}^{N}-\rho_{L}(u_{k}-\mathfrak{m}_{k,n,r}^{N})\right)\Big\}\ .\label{eq:q_acc}
\end{align}

\item Output $u(M)$ as a sample from $\mathcal{K}_{n,r}(\mathfrak{u},\cdot)$. 
\end{enumerate}
\caption{A $\mu_{n,r}$-invariant MCMC Mutation kernel $\mathcal{K}_{n,r}(\mathfrak{u},\cdot)$ }

\label{alg:mcmc_adapted} 
\end{algorithm}

The proposed MCMC kernel is presented in Algorithm \ref{alg:mcmc_adapted}.
For simplicity, in the notation we omit subscripts $n,r$ when writing
$u(m),u$ for $u_{n,r}(m),u_{n,r}$. We also use subscripts $L$,
$H$ to refer to collection of concatenated vectors of real/imaginary
parts of Fourier coefficients in $\mathbf{K}\cap\mathbb{Z}_{\uparrow}^{2}$
and $\mathbf{K}^{c}\cap\mathbb{Z}_{\uparrow}^{2}$ respectively and
$I$ is a $2\times2$ unit matrix. Notice, that even with adaptation,
a few MCMC iterations (denoted by $M\ge1$) might be required to generate
enough jittering for the populations of particles (e.g. $10-30$).
Another interesting feature is the different choice of step sizes
$\rho_{L}$ and $\rho_{H}$ used inside and outside the window respectively.
These step sizes might be still close to $1$ when the target posteriors
admit a very complex correlation structure. Nonetheless, in Section
\ref{sec:Numerical-Examples} we will present numerical examples where
substantially lower values 
can be used to produce reasonable acceptance ratios ($0.15-0.35$)
than when only proposals from the prior are used for all the frequencies
$k\in\mathbb{Z}_{\uparrow}^{2}$. The increased exploration steps
will contribute significantly towards more efficiency in the Monte
Carlo sampling procedure. 

\begin{rem}

In the Navier-Stokes dynamics (\ref{eq:odeNS}), for appropriate choice
of $f,u$ the spectrum of the time evolving vector field $v$ tends
to concentrate at low frequencies. This is common for many dissipative
PDEs. Loosely speaking, the dissipation will transfer energy from
high to low frequencies until $v$ enters the attractor of the PDE.
In cases where the likelihood is formed by noisy observations of $v$,
the resulting posteriors will tend to be more informative at lower
frequencies. Finally, a good initial choice of $K$ could be implied
by the spacing of the observation grid, but this is a heuristic specific
to Eulerian data assimilation.

\end{rem}

\begin{rem}

The acceptance probability in \eqref{eq:acc_ratio} follows from the
usual Metropolis-Hastings accept-reject principle. Compared with the
acceptance probability in \eqref{eq:acc_ratio_simple} the terms involving
$\mathcal{Q},\mu_{0}$ arise to account for the different proposal
at the low frequencies. If we used the proposal in \eqref{eq:proposal_rw_post}
for all $k\in\mathbb{Z}_{\uparrow}^{2}$, it would be necessary to
ensure somehow that the acceptance ratio is well defined and non-zero.
In this case, one should make sure that the corresponding infinite
sums in \eqref{eq:mu_0_acc}-\eqref{eq:q_acc} are finite. A related
practical issue is that any mismatch of the proposed samples with
$\mu_{n,r}$ due to Monte Carlo error can result to many rejected
samples. Using a finite $K$ deals with these issues and one might
also need to use higher values of $N$ when $K$ increases.

We note here that we could opt to use a `smooth' window for $\mathbf{K}$
instead of a window with a sharp edge like (\ref{eq:K}). A possible
implementation could be to modify the proposal to: 
\[
\widetilde{u}_{k,n,r}=\tilde{\mathfrak{m}}_{k,n,r}^{N}+\rho(u_{k}(m)-\tilde{\mathfrak{m}}_{k,n,r}^{N})+\sqrt{1-\rho^{2}}\,\,\mathcal{N}(0,\tilde{\Sigma}_{k,n,r}^{N})\ ,
\]
 where 
\[
\tilde{\mathfrak{m}}_{k,n,r}^{N}=a_{k}\mathfrak{m}_{k,n,r}^{N}\ ,\quad\tilde{\Sigma}_{k,n,r}^{N}=b_{k}\Sigma_{k,n,r}^{N}+(1-b_{k})\tfrac{1}{2}\beta^{2}|k|^{-2\alpha}\, I\ ,
\]
 with $\left(a_{k},b_{k}\right)$ being user-specified sequences converging
to $0$ as $|k|\rightarrow0$. The impact of different choices of
the rate of decay of these sequences is yet to be investigated, but
this is beyond the scope of this paper.

\end{rem}

\subsection{The Complete Algorithm}

\begin{algorithm}
\vspace{0.1cm}

\begin{itemize}
\item At $n=0$, for $j=1,\ldots N$ sample $u_{0,0}^{j}\sim\mu_{0}$. 
\item For $n=0,\ldots,T$

\begin{enumerate}
\item For $j=1,\ldots N$ compute forward dynamics $G_{\delta}\circ G_{\delta}^{(n-1)}(u_{n,0}^{j})$
and $l_{n}(y_{n},u_{n,0}^{j})$. 
\item Set $r=0$ and $\phi_{n,0}=0$. 
\item While $\phi_{n,r}<1$

\begin{enumerate}
\item Increase $r$ by $1$. 
\item (Compute temperature)\\
 IF $\min_{\phi\in(\phi_{n,r-1},1]}\mbox{ES}\mbox{S}_{n,r}(\phi)>N_{thresh}$,
set $\phi_{n,r}=1,$ \\
 ELSE compute $\phi_{n,r}$ such that 
\[
\mbox{ES}\mbox{S}_{n,r}(\phi_{n,r})\thickapprox N_{thresh}
\]
 using a bisection on $(\phi_{n,r-1},1]$. 
\item (Compute weight) for $j=1,\ldots N$: 
\[
W_{n,r}^{j}=\frac{l_{n}(y_{n};u_{n,r-1}^{j})^{\phi_{n,r}-\phi_{n,r-1}}}{\sum_{s=1}^{N}l_{n}(y_{n};u_{n,r-1}^{s})^{\phi_{n,r}-\phi_{n,r-1}}}.
\]
 
\item (Compute moment estimates $\mathfrak{m}_{k,n,r}^{N}$ and $\Sigma_{k,n,r}^{N}$)
for $k\in\mathbf{K}\cap\mathbb{Z}_{\uparrow}^{2}$: 
\[
\mathfrak{m}_{k,n,r}^{N}=\sum_{j=1}^{N}W_{n,r}^{j}u_{k,n,r-1}^{j},\quad\Sigma_{k,n,r}^{N}=\sum_{j=1}^{N}W_{n,r}^{j}\left(u_{k,n,r-1}^{j}-\mathfrak{m}_{k,n,r}^{N}\right)\left(u_{k,n,r-1}^{j}-\mathfrak{m}_{k,n,r}^{N}\right)'.
\]

\item (Resample) let $\left(p_{n}^{1},\ldots,p_{n}^{N}\right)\sim\mathcal{R}(W_{n}^{1},\ldots,W_{n}^{N})$.
For $j=1,\ldots,N$ set $\breve{u}_{n,r}^{j}=u_{n,r-1}^{p_{n}^{j}}$
and $W_{n,r}^{j}=\frac{1}{N}$. 
\item ($\mu_{n,r}$-invariant mutation) for $j=1,\ldots N$, sample $u_{n,r}^{j}\sim\mathcal{K}_{n,r}(\breve{u}_{n,r}^{j},\cdot)$
(see Algorithm \ref{alg:mcmc_adapted}). 
\end{enumerate}
\item Set $u_{n+1,0}^{j}=u_{n,r}^{j}\left(:=u_{n}^{j}\right)$. 
\end{enumerate}
\end{itemize}
\caption{A SMC algorithm for High-Dimensional Inverse Problems}

\label{alg:smc_adapted} 
\end{algorithm}

The complete SMC algorithm is presented in Algorithm \ref{alg:smc_adapted}.
The proposed algorithm computes the temperatures and the scaling of
the MCMC steps on-the-fly using the evolving particle populations.
To improve the mixing of the mutation steps, we use the new adaptive
MCMC kernel in Algorithm \ref{alg:mcmc_adapted} and distinguish between
high and low frequencies. After the completion of Step 4 of Algorithm
\ref{alg:smc_adapted}, the particles can be used to approximate the
intermediate posteriors $\mu_{n}$; this is emphasized in Step 4 by
denoting $u_{n,r}^{j}=u_{n}^{j}$ when $\phi_{n,r}=1$. In addition,
when $\phi_{n,r}=1$ the resampling steps can be omitted whenever
$\mbox{ESS}>N_{thresh}$, as mentioned in Section \ref{sub:A-generic-SMC}.
Although this is not taken into account in Algorithm \ref{alg:smc_adapted},
this extension involves storing (when $\phi_{n,r}=1$) each individual
weight from Step 3(c) as $W_{n}^{j}$, so that at the subsequent weighting
step (only) they can get multiplied to the un-normalized weight: $W_{n+1,1}^{j}\propto W_{n}^{j}l_{n+1}(y_{n+1},u_{n+1,0}^{j})^{\phi_{n+1,1}}$
(similarly to Step 1(a) of Algorithm \ref{alg:basic_smc}).

Although standard SMC methods have a well-developed theoretical framework
for their justification, this literature is much less developed in
the presence of the critical adaptive steps considered in Algorithm
\ref{alg:smc_adapted}. When both adaptive scaling of MCMC steps and
adaptive tempering are considered together, we refer the reader to
\citet{ajay13adaptive} for asymptotic convergence results (in $N$).

As regards the computational cost, for each MCMC mutation at iteration
$n,r$, we need to run the PDE numerical solver $M$ times from $t=0$
up to the current time $t=n\delta$. Therefore, the total computational
cost is proportional to $\kappa MNT^{2}$, where $\kappa$ depends
on the random number of tempering steps. In fact, this cost is significantly
reduced when more tempering steps are required for small values of
$n$. Finally, the memory requirements are $\mathcal{O}(N)$ as there
is no need to store past particles at each step of the algorithm.

\subsection{Monitoring the performance: towards an automatic design}

\label{sec:automatic}

Although SMC is a generic approach suitable for a wide class of problems,
this flexibility also means that the user has to select many design
elements. Firstly, one must decide the target sequence $(\mu_{n})_{n=0}^{T}$,
which often is imposed by the problem at hand as in (\ref{eq:smc_sequence}).
If such a natural sequence is not available, adaptive tempering can
be used directly on $\mu$ as explained in Remark \ref{remark:temper}.
The remaining design choices here are $K,M,N$ and $\rho_{H},\rho_{L}$.
The choice of $N$ can be determined based on the available computational
power and memory. Assuming the user can afford repeated algorithmic
runs, then a useful rule of thumb is to avoid increasing $N$ after
the resulting estimates do not change much. We proceed by outlining
different measures of performance for each of the remaining tuning
parameters: 
\begin{enumerate}
\item For $\rho_{H},\rho_{L}$, we recommend that they should be tuned so
that at time $T$ the acceptance ratio in (\ref{eq:acc_ratio}) averaged
over the particles has a reasonable value, e.g. $0.1-0.3$. Having
said that, the average value of the acceptance ratio should be recorded
and monitored for the complete run of the SMC. The same applies for
the ESS, which will also reveal how much tempering is used during
the SMC run. 
\item It is critical to choose $M$ so as to provide sufficient diversity
in the particle population. A question often raised when using SMC
samplers is whether a given value of $M$ (or a particular MCMC mutation
more generally) is adequate. We propose the following measure for
monitoring the jitter in the population for each frequency $k$ and
each intermediate step of the algorithm $n,r$: 
\begin{equation}
J_{k,n,r}=\frac{\sum_{j=1}^{N}\left|u_{k,n,r}^{j}(M)-u_{k,n,r}^{j}(0)\right|^{2}}{2\sum_{j=1}^{N}\left|u_{k,n,r}^{j}(0)-\mu_{k,n,r}^{N}\right|^{2}}\ .\label{eq:jitter}
\end{equation}
 We remark that the statistic $J_{k,n,r}$ has not been used before
for SMC to the best of our knowledge. Of course, monitoring every
value of $k$ is not necessary, and one could in practice choose a
small number of representative frequencies from their complete set.
In addition, as $N$ increases, $J_{k,n,r}$ will converge to $1-\mbox{corr}(u_{k,n,r}(M),u_{k,n,r}(0))$.
Hence, statistical intuition can explain what requirements we should
pose for $J_{k,n,r}$, e.g. that it should be at least above $0.01-0.05$.
In our context the MCMC mutation steps are applied to jitter the population
at each of the many steps of the SMC sampler. The role of the MCMC
steps here is very different than in a full MCMC sampler, where the
number of MCMC steps have to be large enough for ergodic theory to
apply. 
\item It is also possible to empirically validate whether the chosen value
for $K$, the half-width of $\mathbf{K}$, was appropriate. A revealing
plot here is the two-dimensional heat map of the ratio of the variance
of the posterior to the variance of the prior against $k$. $\mathbf{K}$
should include as much as possible the region where this ratio is
significantly less than $1$, e.g. say less than $0.8$. 
\end{enumerate}
All the above criteria can be used to empirically evaluate algorithmic
performance during or after a SMC run with a particular set of tuning
parameters. Fortunately, SMC does not require complicated convergence
diagnostics similar to MCMC.

A particularly interesting point is that summaries like the average
acceptance ratio, the statistic $J_{k,n,r}$, or the ratio of the
posterior to prior variance can be also be potentially used within
decision rules to determine $\rho_{L},\rho_{H}$, $M$ and $K$ adaptively
on-the-fly. This can be implemented similarly to how the ESS is used
for adaptively choosing the tempering temperatures. To choose $M$
for instance one may keep using MCMC iterations until the median of
$J_{k,n,r}$ over a number of representative frequencies $k\in\mathbb{Z}_{\uparrow}^{2}$
reaches a pre-specified value. For $\rho_{L},\rho_{H}$ an option
can be to increase/decrease them by a given fraction according to
the values of the empirical average acceptance ratio over the particle
population as done in \citet{jasra2011inference}. Similarly, $K$
can vary according to the ratio of the estimated posterior variances
to the prior ones on the boundary of $\mathbf{K}$. 
These aforementioned guidelines or other similar ideas can lead towards
a more automatic implementation of the proposed SMC algorithm. Although
we believe investing effort towards this direction is definitely important,
this lies beyond the scope of this paper. So, in the numerical examples
to follow the performance measures mentioned here are used as monitoring
tools but not as means to adapt more tuning parameters.

\begin{rem} We should clarify that adaptation in this paper is of
a very different nature than what is commonly referred to as adaptive
MCMC in the literature, e.g. \citet{andrieu2008tutorial}. There,
the MCMC tuning parameters such as $\rho$ in Algorithm \ref{alg:mcmc}
vary according to the history of the sample path of the MCMC trajectory.
This method is intended for long MCMC runs and care needs to be taken
to ensure that appropriate diminishing adaptation conditions are satisfied.
Here the interest is only to jitter the population. The MCMC mutations
have different invariant measures at each $n,r$ and are based on
very few standard MCMC iterations, which conditional on the whole
particle population are independent for each particle.

\end{rem}

\section{Numerical Examples}

\label{sec:Numerical-Examples}

We will use numerical solutions for the Navier-Stokes PDE in \eqref{eq:odeNS}
given an initial condition. We employ a method based on a spectral
Galerkin approximation of the velocity field in a divergence-free
Fourier basis (\citet{hesthaven2007spectral}). The convolutions arising
from products in the nonlinear term are computed via FFTs on a $64^{2}$
grid with additional anti-aliasing using double-sized zero padding
(\citet{uecker2009short}). In addition, exponential time differencing
is used as in \citet{cox2002exponential}, whereby an analytical integration
is used for the linear part of the PDE, $\nu Av$, together with an
explicit numerical Euler integration scheme for the non-linear part,
$P\left(f\right)-B(v,v)$.

In this section we will present two numerical examples using two different
synthetic data-sets obtained from some corresponding true initial
vector fields, $u^{\dagger}$. The first example will use Data-set
A consisting of few blocks of observations obtained at close time
intervals (small $\delta$), but each block is obtained from a dense
observation grid (high $\Upsilon$). For this example we will compare
the performance of our method with benchmark results from the MCMC
approach described in Algorithm \ref{alg:mcmc}. For the second example
we will use Data-set B, a longer data-set with blocks of observations
spaced apart by longer time periods $\delta$, each originating from
a sparser observation grid (low $\Upsilon$). In both cases the total
number of observations of the vector field will be the same and $\Upsilon T=80$.
We summarize these details in Table \ref{tb:data}. 
\begin{table}[!h]
\centering %
\begin{tabular}{|c|c|c|}
\hline 
Data-set A  & $\delta=0.02$\ ,\,\,\,$\Upsilon=16$\ ,\,\,\, $T=5$  & $u^{\dagger}\sim\mathcal{N}(0,\,\beta^{2}\, A^{-\alpha})\,,\,\,\beta^{2}=5,\,\,\alpha=2.2$ \tabularnewline
\hline 
Data-set B  & $\delta=0.2$ \ ,\,\,\, $\Upsilon=4$\ ,\,\,\, $T=20$  & $u^{\dagger}\sim\mathcal{N}(0,\,\beta^{2}\, A^{-\alpha})\,,\,\,\beta^{2}=1,\,\,\alpha=2$ \tabularnewline
\hline 
\end{tabular}\caption{The specification of the two datasets considered in the numerical
examples. Data-set A corresponds to a scenario of a short-time data-set
with a dense observation grid and Data-set B to the scenario of a
long data-set with a sparse observation grid. The true initial condition
$u^{\dagger}$ was sampled from the prior with the shown parameter
values.}

\label{tb:data} 
\end{table}

The two data-sets are synthesized using the numerical PDE solver described
above. In both cases we have set: 
\[
\nu=0.02\ ,\quad f(x)=\nabla^{\perp}\cos\big((5,5)^{'}\cdot x\big)\ ,\quad\gamma^{2}=0.2\ ,
\]
 where we remind the reader that $\nu$ is the viscosity, $f$ the
external forcing, and $\gamma^{2}$ the observation noise variance.
Adjoint PDE solvers will be used in the sense that the same numerical
solver is used for synthesizing the data and in the Monte Carlo inference
algorithms. 

\begin{table}[!h]
\centering%
\begin{tabular}{|c|c|c|c|}
\cline{2-4} 
\multicolumn{1}{c|}{} & Algorithmic  & number of calls of PDE solver  & execution time\tabularnewline
\multicolumn{1}{c|}{} & Parameters  & with time length $\delta$ (divided by $T$)  & \tabularnewline
\hline 
MCMC Dataset A  & $\rho=0.9998$  & $0.9\times10^{6}$  & 9 days\tabularnewline
\hline 
SMC Dataset A  & $N=500\,,\, M=20\,$  & \multirow{2}{*}{$7.266\times10^{5}$}  & \multirow{2}{*}{3 days}\tabularnewline
no parallelization  & $\rho_{L}=0.99\,,\,\,\rho_{H}=0.991\,,\,\, K=7$  &  & \tabularnewline
\hline 
SMC Dataset A  & $N=1,020\,,\, N_{thresh}=\frac{N}{3},\, M=20\,$  & \multirow{2}{*}{$1.403\times10^{6}$}  & \multirow{2}{*}{7.4 hours}\tabularnewline
with parallelization  & $\rho_{L}=0.99\,,\,\rho_{H}=0.991\,,\,\, K=7$  &  & \tabularnewline
\hline 
SMC Dataset B  & $N=1,020\,,\, N_{thresh}=\frac{N}{3}\,,M=20\,$  & \multirow{2}{*}{$1.447\times10^{6}$}  & \multirow{2}{*}{3.5 days}\tabularnewline
with parallelization  & $\rho_{L}=0.99\,,\,\rho_{H}=0.991\,,\,\, K=7$  &  & \tabularnewline
\hline 
\end{tabular}\caption{We present the number of times a numerical PDE solution of total length
$\delta$ is required by each algorithm. This number is divided by
$T$. The total execution time is also shown for each case. For the
parallel implementation of Algorithm \ref{alg:smc_adapted} we used
trivial parallelization (except for the resampling step). The code
was written in Matlab$^{(R)}$ and parallel implementations of SMC
run as a parallel MPI job with 60 workers on the computing cluster
of CSML-UCL (SunGrid$^{(R)}$ engine). All other simulations were
performed in Matlab$^{(R)}$ on the same computer running Linux with
an Intel$^{(R)}$ Xeon$^{(R)}$ CPU E5-1660 at 3.30GHz (six core)
and 16 GB RAM.}

\label{tab:comp} 
\end{table}

In Table \ref{tab:comp} we summarize the computational cost of the
algorithms used in our experiments. The table presents the number
of times a PDE solution is required and the total execution time.
We do not provide an MCMC benchmark for Data-set B as the more expensive
PDE solver needed in this case would result in an enormous execution
time for Algorithm \ref{alg:mcmc}. For SMC, we will present results
from an implementation of Algorithm \ref{alg:smc_adapted} with trivial
parallelization, whereby the resampling step is performed at a single
computing node that collects and distributes all particles. For Data-set
A, in Table \ref{tab:comp} we also show the computational cost from
a typical run of SMC with $N=500$ but without using parallelization.
Although in the remainder of this section we will not present the
actual results from this run, we report that the performance was comparable
to MCMC as well as SMC with higher $N$ obtained via the parallel
implementation. In a way this demonstrates the efficiency of the SMC
method compared to MCMC, but we have to emphasize that for more realistic
applications parallelization is critical for effective execution times.

\subsection{Data-set A: Short-time Data-set with a Dense Observation Grid}

Figure \ref{ex1:post_mean} plots the posterior mean of the vorticity
and velocity fields for the initial condition as estimated by our
adaptive SMC algorithm (left) and the MCMC one (right). In the same
plot the true field $u^{\dagger}$ and its vorticity is displayed
in the middle. The results from SMC and MCMC are very similar and
both methods manage to capture the main features of the true field
$u^{\dagger}$. The smoothing effect observed for the posterior means
by both SMC and MCMC appears because the observations cannot provide
substantial information about the high frequency Fourier coefficients.
Figure \ref{ex1:pred} shows the posterior mean of the vorticity field
this time for the terminal state $v(\cdot,\delta T)$ instead of the
initial condition (i.e. the push forward probability measure of $\mu$
under $G_{\delta}^{(T)}$).

Note that the objective here is to approximate the full posterior
and not just the mean. Figure \ref{ex1:pdfs} shows the estimated
posterior density functions (PDFs) for a few (re-scaled) Fourier coefficients,
$\xi_{k,T}$ (as defined in \eqref{eq:kl1-1}), of different frequencies
$k$ obtained using both SMC and MCMC. Recall the scaling in $\xi_{k,0}$
makes all prior densities equal to standard normals. In Figure \ref{ex1:pdfs}
we plot the prior density (dotted), the posterior densities from SMC
(solid) and MCMC (dashed) together with the true value of $\xi_{k,T}^{\dagger}$
used for generating the data (vertical line). In Figures \ref{ex1:scatterMCMC}
and \ref{ex1:scatterSMC} we show scatter-plots over pairs of different
frequencies $k$ from the (sub-sampled) MCMC trajectory and the SMC
particles respectively. In all the aforementioned plots SMC and MCMC
seem to be in close agreement, which provides numerical evidence that
SMC samples correctly from the posterior distribution.

We proceed by presenting different measures of performance for MCMC
and SMC. In Figure \ref{ex1:monitorMCMC} we plot the autocorrelations
of the MCMC trajectory for different Fourier coefficients. The mixing
of the MCMC chain is quite slow, hence a large number of iterations
was required for the MCMC approach to deliver reliable results. To
monitor the performance of the SMC algorithm, Figure \ref{ex1:monitorSMC}
includes plots of the ESS, the average acceptance ratio and the jittering
indicator $J_{k,n,r}$ against each SMC iteration%
\footnote{By SMC iteration index $n,r$ we mean actually iteration $r+\sum_{p=0}^{n-1}q_{p}$
, i.e. the number of times Steps 1-4 of Algorithm \ref{alg:smc_adapted}
have completed. Note $q_{p}$ is as in \eqref{eq:number_freq} and
is a random variable determined by the algorithm.%
}$n,r$. Compared to the size of the data-set ($\Upsilon T=80$) the
total number of extra tempering steps required here was about $50$.
In the bottom left plot in Figure \ref{ex1:monitorSMC} we observe
how the average, maximum and minimum (over $k$) of $J_{k,n,r}$ changes
with $n,r$ separately for when $k$ is within or outside the window
of frequencies $\mathbf{K}$. For some indicative values of $k$ we
also show $J_{k,n,r}$ in the lower right plot of Figure \ref{ex1:monitorSMC}.
$J_{k,n,r}$ does not seem to vary a lot with $\left|k\right|$. It
is certainly reassuring that all the MCMC steps seem to deliver considerable
amount of jittering to all Fourier coefficients. Also, the amount
of jittering appears to be fairly evenly spread over all Fourier coefficients,
even if different MCMC proposals are used within and outside the window
$\mathbf{K}$. Although supporting plots are not shown here, $J_{k,n,r}$
seemed to grow linearly with $M$ for every $k$ and each $n,r$.


In Figure \ref{ex1:circle} we examine some statistical properties
that are related to frequencies $k$. In this plot, in the top row
we use a heat map against $k$ to plot the ratio of estimated posterior
marginal standard deviation for $\xi_{k,T}$ over the standard deviation
of the $\xi_{k,0}$ (or the prior). The left plot corresponds to this
ratio computed using the real parts of $\xi_{k,T}$ and the right
plot to imaginary ones. In the bottom row we plot the posterior mean
of $\xi_{k,T}$ against $k$. We can deduce that the choice of treating
high and low frequencies separately and adapting the MCMC steps only
for a window of frequencies is reasonable, as most of the information
in the data spreads over a number of frequencies covered by our chosen
window $\mathbf{K}$.

\subsection{Dataset B: Long-time Data-set with a Sparse Observation Grid}


We will now present the results of the SMC method of Algorithm \ref{alg:smc_adapted}
when applied to Data-set B. This scenario is more challenging than
the one with Data-set A, as we allow the Navier-Stokes dynamics to
evolve for a longer period of time. We will follow a similar presentation
as in the previous example. Figure \ref{ex2:post_mean} plots the
posterior mean of the initial vorticity and velocity field. This plot
can be used to compare the SMC estimates method versus the true values
corresponding to $u^{\dagger}$. Although here we do not have a benchmark
available like before, the smoothing effect in the estimates relative
to the truth does not seem surprising based on the intuition gained
from the previous example. In Figure \ref{ex2:pred} we show the posterior
mean of the final vorticity $\varpi(\cdot,\delta T)$ together with
the true one. Although few errors are present, the result is satisfying
given the (mildly) chaotic nature of the forward dynamics.

Figure \ref{ex2:pdfs} displays the approximate posterior densities
of $\xi_{k,n}$ for a number of frequencies $k$. The difference compared
to Figure \ref{ex1:pdfs} for the previous example is that now we
can see how $\mu_{n}$ changes for $n=0\mbox{ (dotted)},$ $0.5T\mbox{ (dash-dotted)},$
$\:0.75T\mbox{ (dashed)},$ $\: T\mbox{ (solid)}$. As expected, each
new block of observations contributes to shaping a more informative
posterior. In Figure \ref{ex2:scatterSMC} we present the scatter-plots.
To monitor the performance of SMC, in Figure \ref{ex2:monitorSMC}
we plot the ESS, average acceptance ratio and $J_{k,n,r}$ all against
$n,r$ as we did for the previous example. The algorithm uses almost
the same number of tempering steps in total compared to the previous
example. In addition, the acceptance ratio and $J_{k,n,r}$ stop decreasing
after some iteration. We interpret this as a sign that $\mu_{n}$
stops changing fast with $n$ and that the particles form good approximations
of the targeted sequence.

Finally, Figure \ref{ex2:circle} shows the heat maps against $k$
of the estimated posterior means of $\xi_{k,T}$ and the ratio of
their marginal posterior standard deviations over their prior values,
similarly to Figure \ref{ex1:circle} for Data-set A. Compared to
the previous example the posterior seems to gain information from
the observations for a wider window of low frequencies. Still, the
choice of $K=7$ seems to be justified.

\section{Discussion and Extensions\label{sec:Discussion-and-extensions}}

This paper aims to make a significant contribution towards challenging
the perception that SMC is not useful for high-dimensional inverse
problems. We believe this appears to be often the case when particle
methods are implemented naively, see for example some negative results
and exponential-in-dimension computational costs reported in \citet{bengtsson2008curse,snyder2008obstacles}
for some cases involving stochastic dynamics. The added efficiency
of our method compared to plain MCMC can be attributed to being able
to employ a variety of adaptation steps that take advantage of the
evolving particle population, hence tuning the algorithm effectively
to the structure of the target distributions in the SMC sequence.
SMC algorithms are also appealing to practitioners given the inherent
ability to parallelize many steps in the algorithm thus drastically
reducing execution times. As regards to understanding the effect of
each block of observations, another useful aspect of the method is
that the SMC sequence allows for monitoring the evolution of posterior
distributions of interest as more observations arrive. In contrast,
MCMC methods would require re-running the algorithm from scratch.
In terms of the accuracy of the estimates we believe that SMC can
be on a par with expensive MCMC methods, and this is illustrated clearly
in the example with Data-set A.

We believe that the numerical results for the case study in this paper
can motivate further investigations for using exact methods for data
assimilation problems. In the present case study, the proposed SMC
algorithm was able to provide results in an example (Data-set B) where
the execution time of MCMC was prohibitive. Of course, PDE models
used in practice can be much more complex that the Navier-Stokes equations
used in this paper and might require a much finer resolution. Therefore,
it would be unrealistic to claim that SMC methods can replace common
heuristics. However, we believe that a thorough exploration of SMC
methods in this context is beneficial for the data assimilation community,
especially if they can be combined with some justifiable approximations
to reduce the computational burden. Recent research initiatives in
this direction include \citet{van2010nonlinear,chorin2010implicit,rebeschini2013can}
when stochastic dynamics are used. In particular, \citet{rebeschini2013can}
look at ways of reducing the intrinsic deficiency of importance sampling
by avoiding to dismiss a whole high-dimensional particle when only
few of its coordinates are in disagreement with the observation's
likelihood. It would be interesting to investigate how the methodology
developed here can be combined with these aforementioned works.

\begin{algorithm}
\vspace{0.1cm}
 
\begin{itemize}
\item At time $\tau=L-1$ execute Algorithm \ref{alg:smc_adapted} for $n=0,\ldots,L-1$.

\begin{itemize}
\item Store particles $\{u_{n}^{j}\}_{j=1}^{N}$ after resampling for all
steps $0<n\le L$. 
\end{itemize}
\item At times $\tau=L,\ldots,T$, execute Algorithm \ref{alg:smc_adapted}
with $n=\tau-L+1,\ldots,\tau$ as follows:

\begin{enumerate}
\item Replace $\mu_{0}$ in Algorithm \ref{alg:mcmc_adapted} with an approximation
of the posterior $u|y_{1},\ldots,y_{\tau-L}$ constructed using the
particle population $\{u_{\tau-L}^{j}\}_{j=1}^{N}$, e.g.\@ $\mathcal{N}(\mu_{\tau-L}^{N},\Sigma_{\tau-L}^{N})$,
or some smooth Gaussian mixture centered at the particle locations
$\{u_{\tau-L}^{j}\}_{j=1}^{N}$. 
\item Store particles $\{u_{n}^{j}\}_{j=1}^{N}$ after resampling over all
steps $\tau-L+1\leq n\le\tau$. 
\item Discard from memory particles $\{u_{\tau-L}^{j}\}_{j=1}^{N}$ . 
\end{enumerate}
\end{itemize}
\caption{Receding Horizon Estimation for Long $T$ with SMC-Sampling}

\label{alg:mpc} 
\end{algorithm}

The work in this paper opens several paths for further investigations
towards developing effective exact algorithms for data assimilation.
In the context of high-dimensional Bayesian inference either for inverse
problems or non-linear filtering we mention some possible extensions
below: 
\begin{description}
\item [{Gaussian}] \textbf{priors and beyond:} We have chosen Gaussian
priors with a fixed rate of decay $\alpha$ for the eigenvalues of
their covariance operators. In general, $\alpha$ should be small
enough so that posteriors for Fourier coefficients for which there
is information in the data are not dominated by the prior, and this
requires further empirical/theoretical investigations. Other prior
distribution options could also be investigated. For instance, in
the context of plain MCMC algorithms, the works in \citet{cotter2012mcmc}
and \citet{dashti2011besov} have considered Sieve and Besov priors
respectively. 
\item [{Algorithmic}] \textbf{design automation:} The added flexibility
of the SMC framework comes at the price of having to choose more algorithmic
tuning parameters (in our context, choice of step-sizes $\rho_{H}$,
$\rho_{L}$, window size $K$, number of MCMC steps $M$). All these
parameters are involved in the specification of the MCMC mutation
step, that must provide enough jittering to the data. As explained
in Section \ref{sec:automatic}, there is a lot of information in
the particle population that can be exploited towards a more `automated'
version of the algorithm presented here, with more (if not all) parameters
determined on-the-fly. %

\item [{Algorithmic}] \textbf{robustness:} We have considered the standardized
squared jumping distance index $J_{k,n,r}$ in (\ref{eq:jitter})
as a way of measuring the diversity in the particles during the execution
of the algorithm. In the case of Data-Set A, when an alternative algorithm
(i.e.\@ plain MCMC) is available for comparison, it appears that
even small values ($\sim5\%$) can suffice. More investigations are
needed to determine what values of $J_{k,n,r}$ should one aim for
in general. 
\item [{On-line}] \textbf{algorithms:} The described algorithm has computational
cost $\mathcal{O}(T^{2})$, so currently is not useful for on-line
applications with very long data-sets. One possible remedy could be
to use receding horizon estimation principles similar to \citet{jazwinski1968limited,rawlings2006particle}
with a fixed lag or memory $L$ to reduce the computational cost to
$\mathcal{O}(L^{2}T)$ with increased memory requirements of $\mathcal{O}(NL)$.
An example of such a procedure is shown in Algorithm \ref{alg:mpc},
where one implements Algorithm \ref{alg:smc_adapted} recursively
with a fixed-lag approach and for the MCMC jittering steps in Algorithm
\ref{alg:mcmc_adapted} one uses instead of $\mu_{0}$ a Gaussian
approximation (or a mixture of them) of the posterior at time $\tau$.
A bias will incur from using these approximations of the posterior
at time $\tau$ in the $L$ subsequent mutations and importance sampling
selections. Hopefully in some cases this bias might be small and uniform
in time. In principle, one could also correct this estimation error
using some form of importance sampling like \citet{doucet2006efficient},
but it remains to be investigated whether this is possible to be implemented
without requiring weight computations looking all the way back to
$\tau=0$ and hence being impractical for long $\tau$. Algorithm
\ref{alg:mpc} bears some similarity with earlier works: in \citet{stordal2011bridging,frei2012bridging,kai12thesis}
the authors combine ensemble Kalman filtering or Gaussian mixtures
with particle filters when the noise is present in the dynamics without
any use of receding horizon or MCMC jittering steps and in \citet{yang2010receding}
the authors present a similar algorithm to Algorithm \ref{alg:mpc}
intended for a computer vision application but again without using
any MCMC mutation steps. We are currently investigating the performance
and theoretical properties of Algorithm \ref{alg:mpc} when applied
to inverse problems involving long observation sequences. 
\item [{Beyond}] \textbf{Navier-Stokes on a Torus:} An advantage of the
proposed SMC methodology is that it is generic and applicable to a
wider class of applications than the particular Navier-Stokes inverse
problem. One could use this method also with non-linear observations,
different dynamical systems such as groundwater flow equations or
other statistical inverse problems such as those examined in \citet{stuart2010inverse,cotter2012mcmc}.
In addition, a challenging extension could be to see how this work
can be useful in the context of high dimensional filtering when noise
is present in the dynamics. 
\end{description}

\section*{Acknowledgements}

The authors would like to thank Dan O. Crisan, Kody J.H. Law and Andrew
M. Stuart for many useful discussions that provided valuable insight
on the topic. We would also like to thank Vassilis Stathopoulos for
his help on using Matlab$^{(R)}$ parallel jobs with MPI on the SunGrid$^{(R)}$
engine of CSML at UCL. N. Kantas and A. Beskos received funding by
EPSRC under grant EP/J01365X/1, which is gratefully acknowledged.
A. Jasra and A. Beskos acknowledge also support from MOE Singapore.

\bibliographystyle{plainnat} \addcontentsline{toc}{section}{\refname}
 \bibliographystyle{plainnat}
\bibliography{navier}
 
\begin{figure}
\centering\includegraphics[width=1\textwidth]{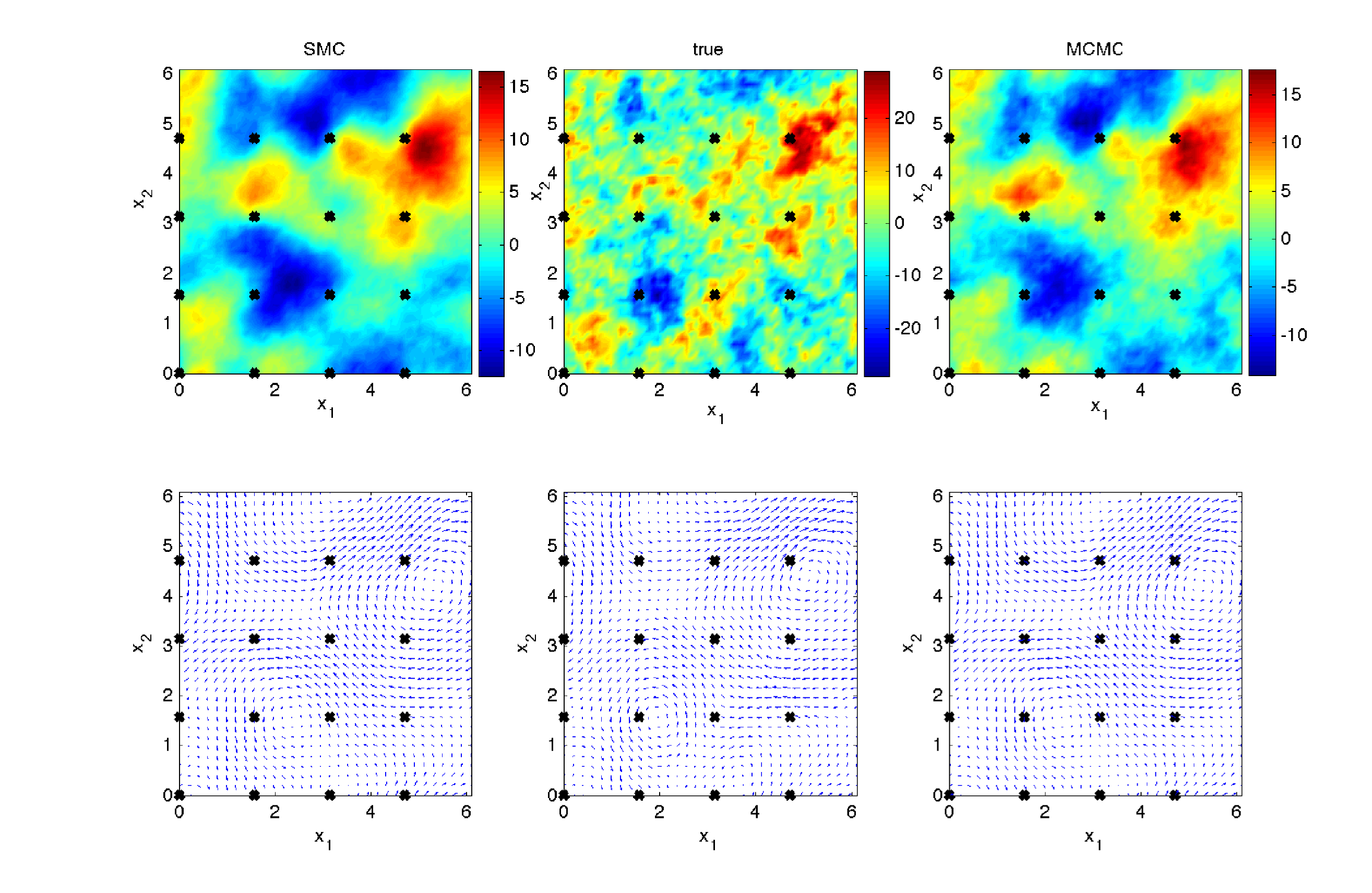} \caption{Data-set A. Top panel: left, posterior mean of initial vorticity from
SMC; center, true initial vorticity; right, posterior mean of initial
vorticity from MCMC. Bottom panel: corresponding graphs for the velocity
fields with same order from left to right. The crosses indicate the
positions $x_{1},\ldots,x_{\Upsilon}$ where the vector field is observed.}

\label{ex1:post_mean} 
\end{figure}

\begin{figure}
\centering\includegraphics[width=1\textwidth]{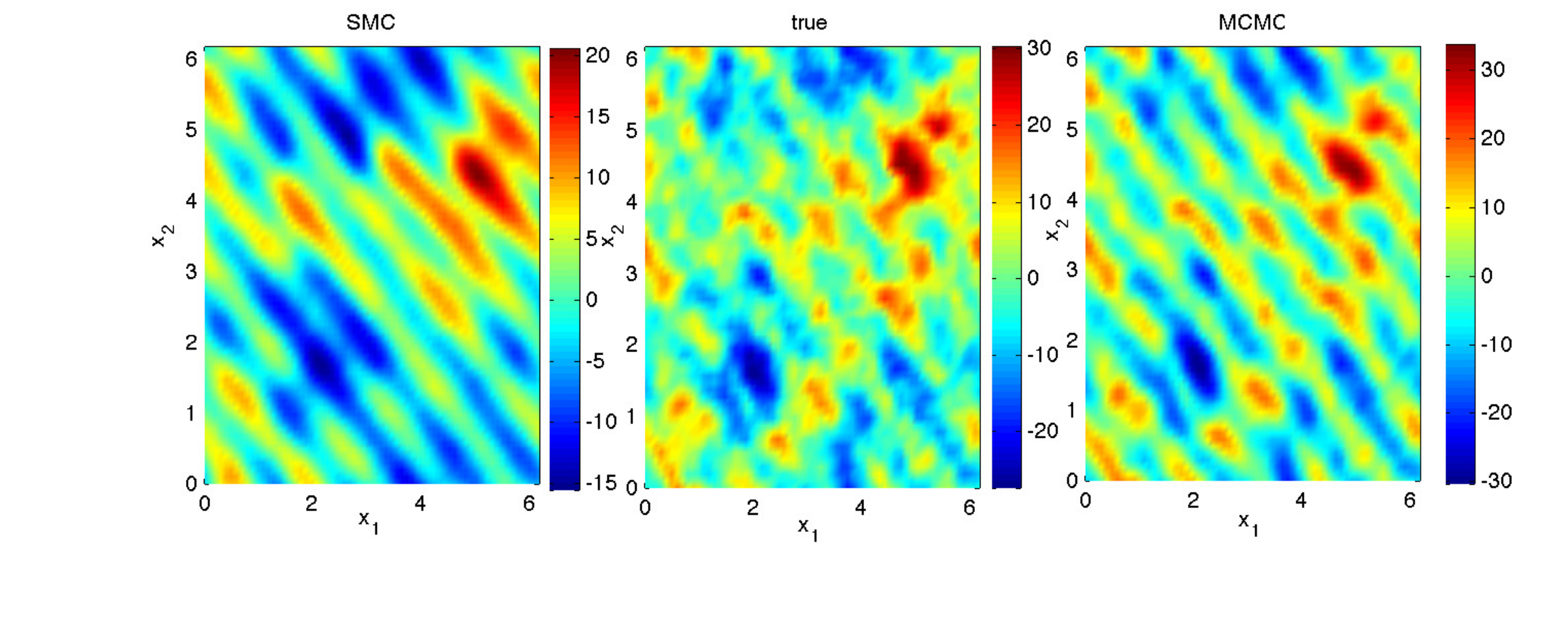} \caption{Data-set A. Vorticity of $v(\cdot,\delta T)$. Posterior mean from
SMC (left) true initial condition (center) and posterior mean from
MCMC (right). }

\label{ex1:pred} 
\end{figure}

\begin{figure}
\centering \hspace{-1cm} \includegraphics[width=1\textwidth]{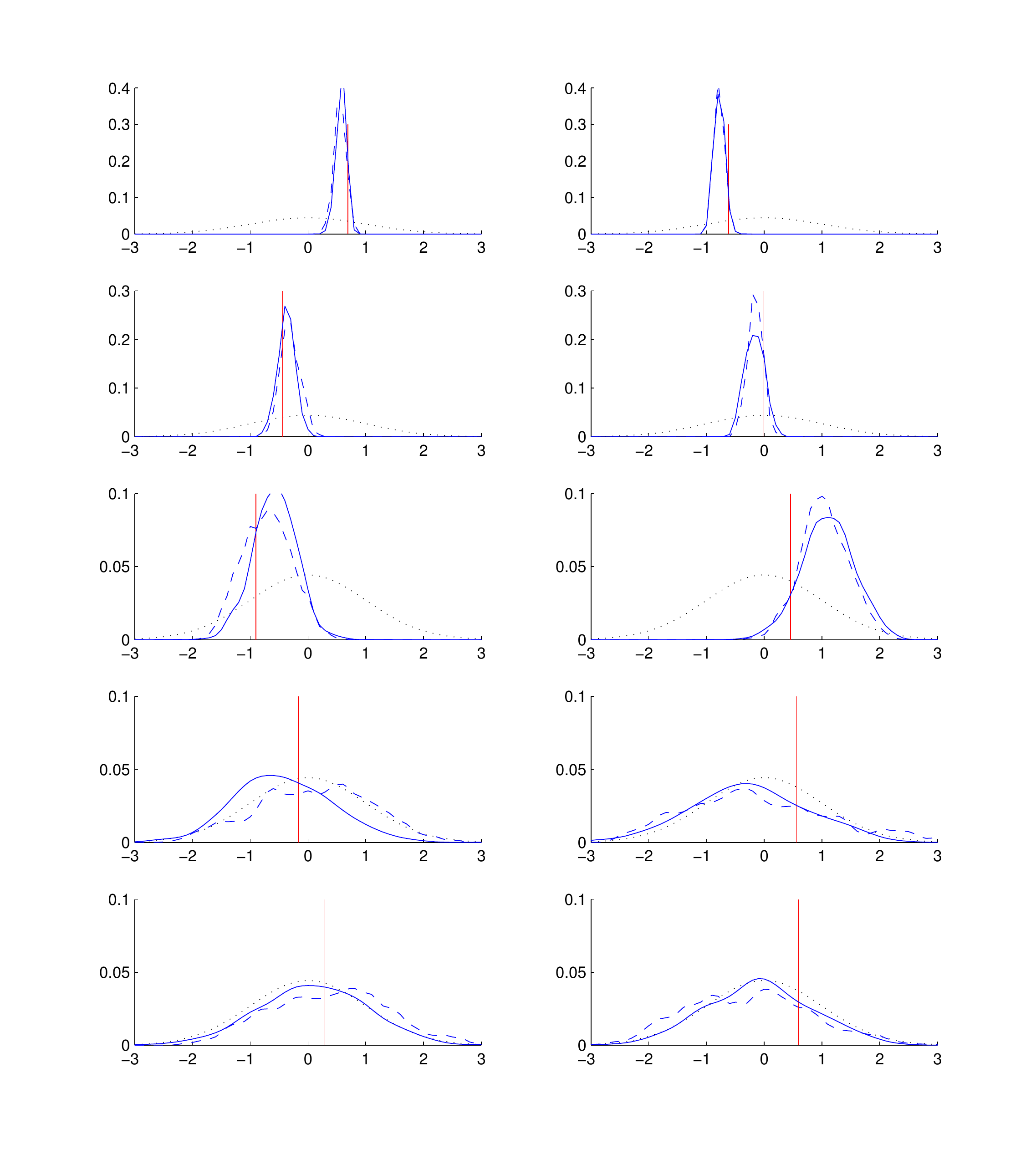}
\vspace{-2cm}
 \caption{Data-set A. Estimated PDFs: Blue lines are the estimated posterior
densities for $\xi_{k,T}$; the solid lines are for SMC, the dashed
ones for MCMC and the dotted (black) lines correspond to the prior
densities. The left panel is for the real parts and right panel for
the imaginary parts. The different rows correspond to each of the
frequencies $k=(0,1),(1,1),(2,1),(4,4),(9,9)$ from top to bottom.
The red vertical lines designate the true values $\xi_{k,T}^{\dagger}$.}

\label{ex1:pdfs} 
\end{figure}

\begin{figure}
\centering \includegraphics[width=1\textwidth]{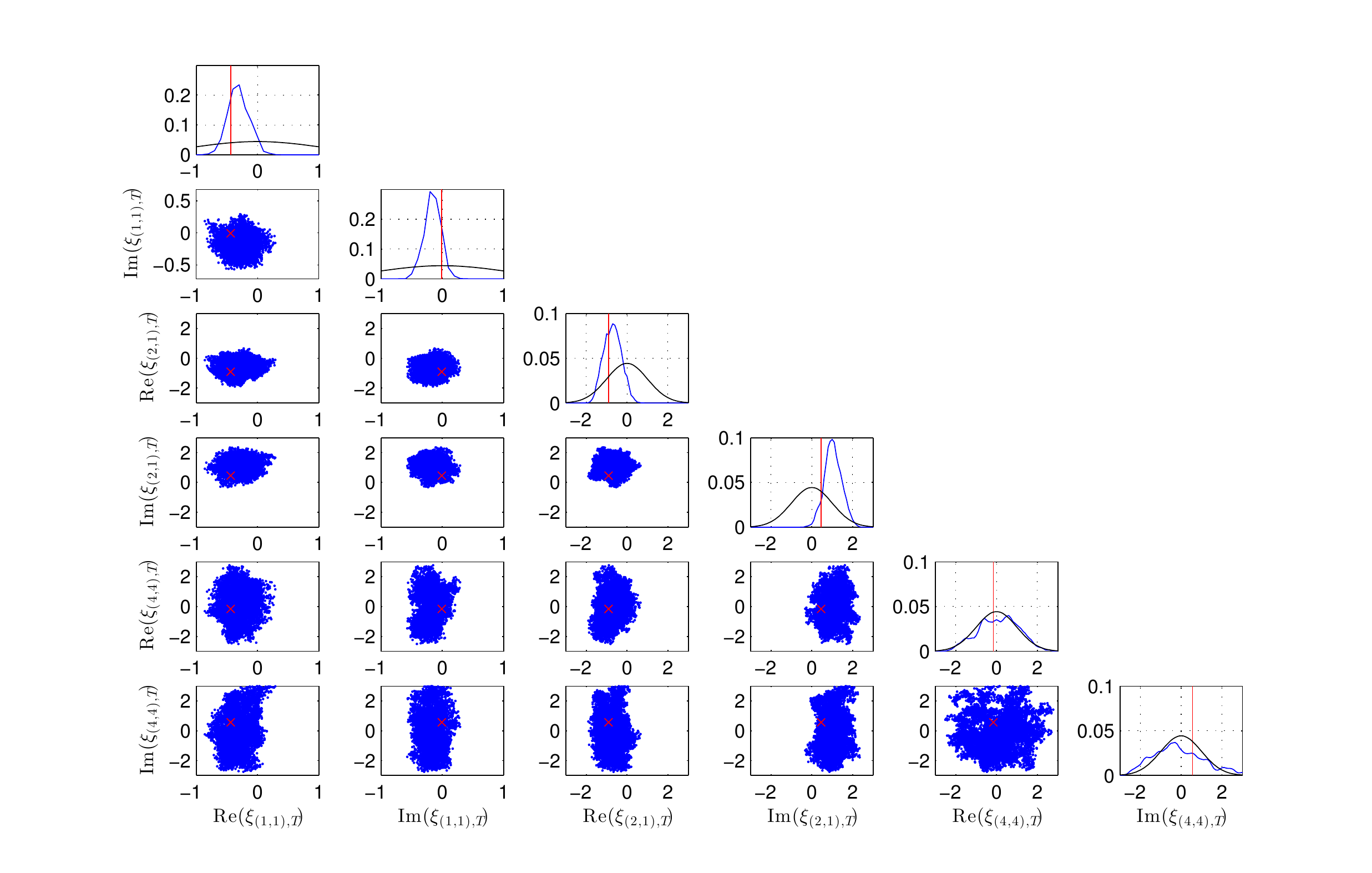} \caption{Data-set A: Scatter-plots generated from the sub-sampled MCMC trajectory
(every 100th iteration). For each row (and column resp.\@) the y-axis
(and x-axis resp.\@) alternates between real and imaginary parts
of $\xi_{k,T}$ for $k=(1,1),(2,1),(4,4)$. The red crosses in the
scatter-plots show the true values $\xi_{k,T}^{\dagger}$. Graphs
in the diagonal show the estimated posterior PDF's in blue together
with the prior, in black, and the true value as the red vertical lines
(as in Figure \ref{ex1:pdfs}). }

\label{ex1:scatterMCMC} 
\end{figure}

\begin{figure}
\includegraphics[width=0.43\textwidth]{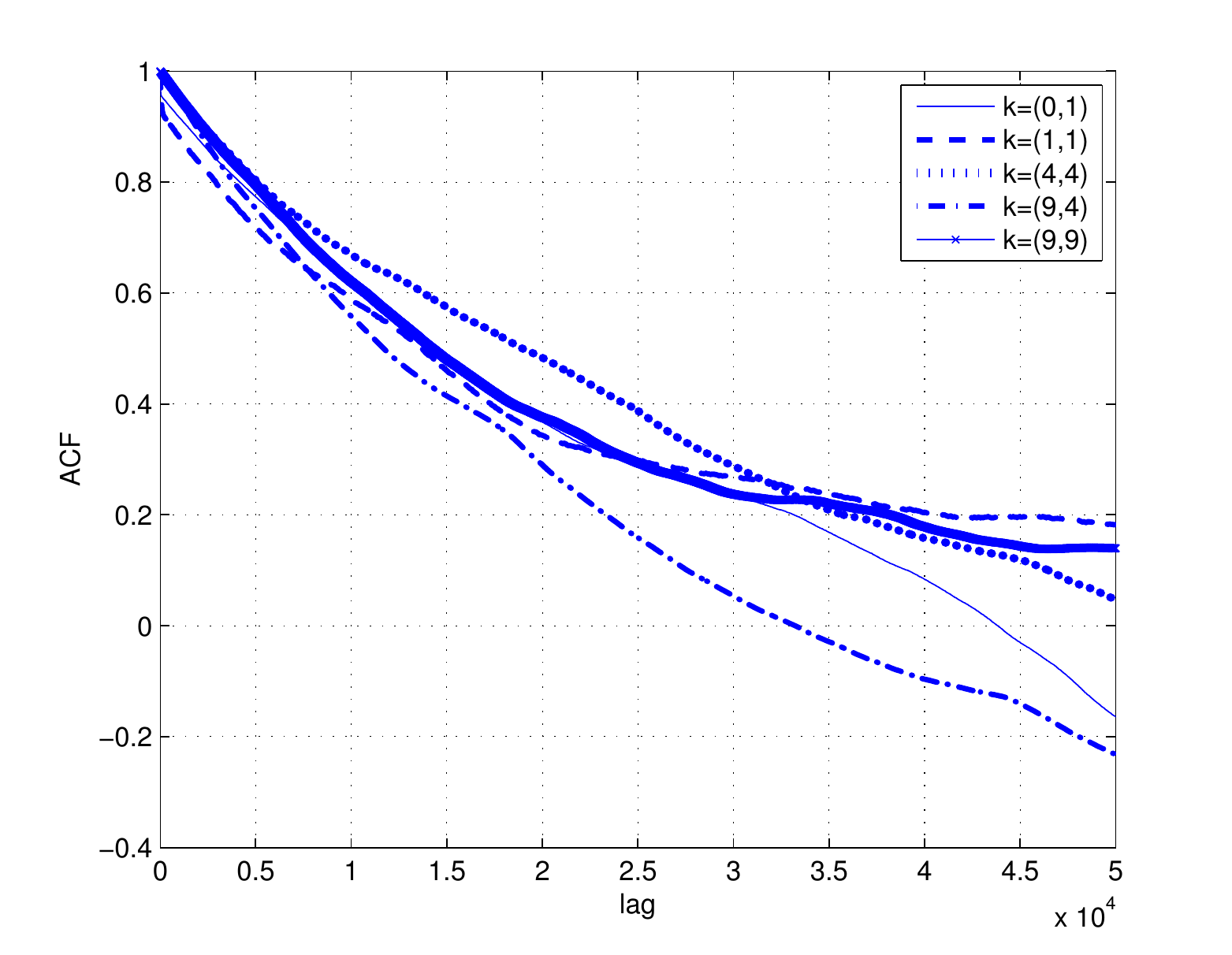}\includegraphics[width=0.45\textwidth]{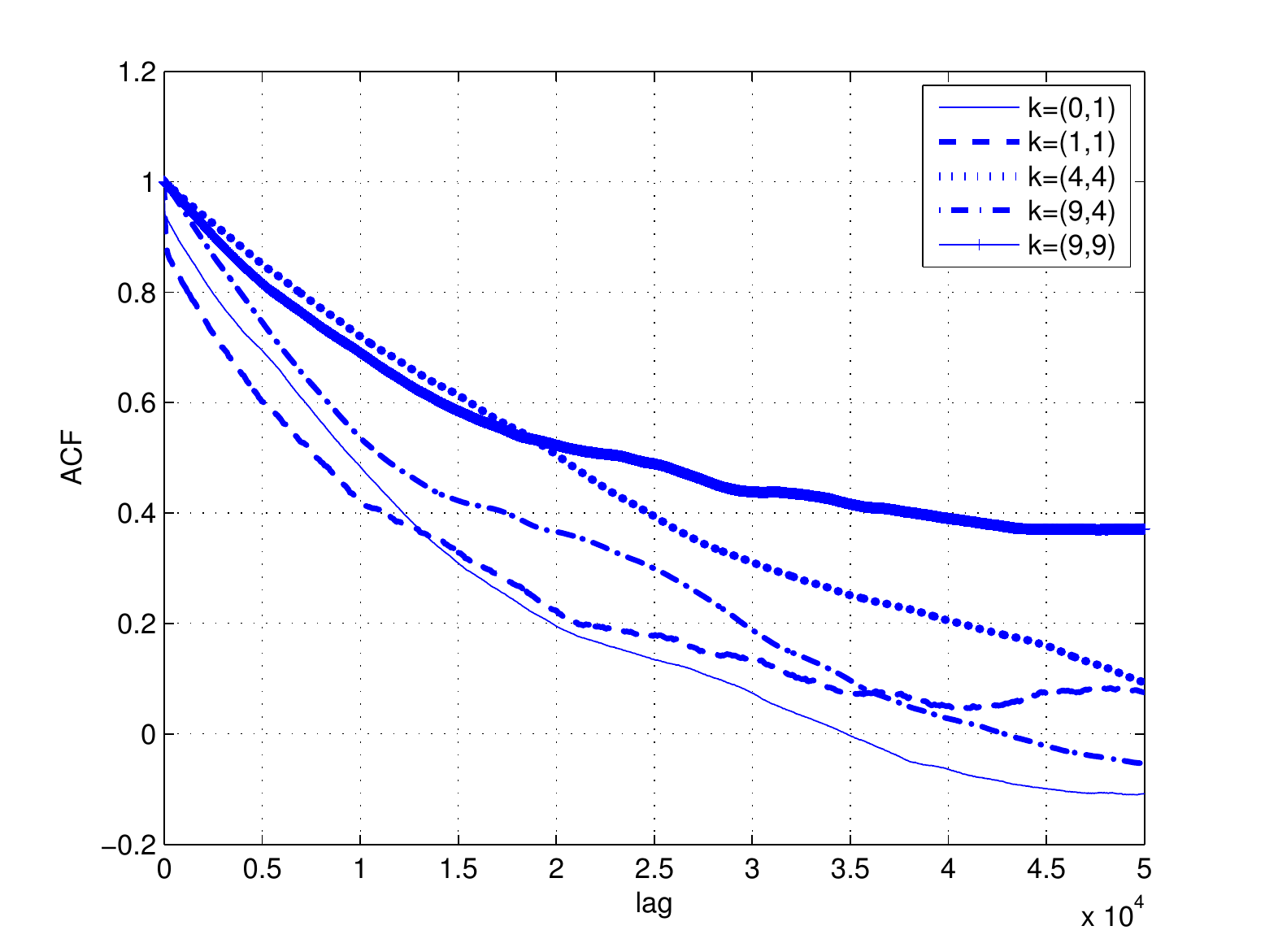}
\caption{Data-set A: monitoring MCMC performance. Autocorrelation plots from
the MCMC trajectory of $\xi_{k,T}$, for a number of different frequencies;
left graph corresponds to the real parts, and the right one to the
imaginary ones.}

\label{ex1:monitorMCMC} 
\end{figure}

\begin{figure}
\centering\includegraphics[width=1\textwidth]{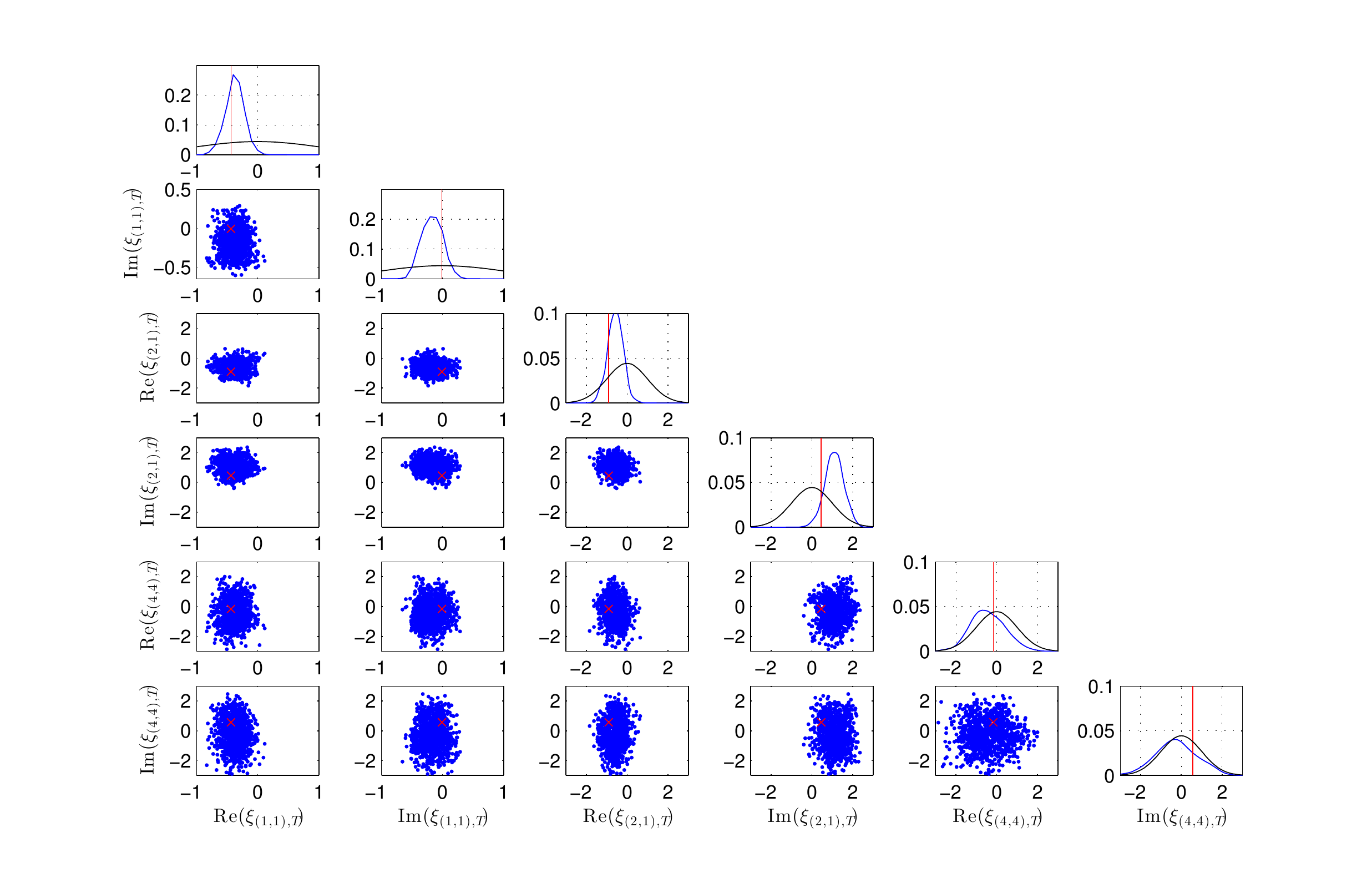}
\caption{Dataset A: Scatter-plots generated from the SMC particle ($N=1020$)
for $\xi_{k,T}$ at frequencies $k=(1,1),(2,1),(4,4)$. Details are
similar to Figure \ref{ex1:scatterMCMC}.}

\label{ex1:scatterSMC} 
\end{figure}

\begin{figure}
\includegraphics[width=0.45\textwidth,height=0.18\textheight]{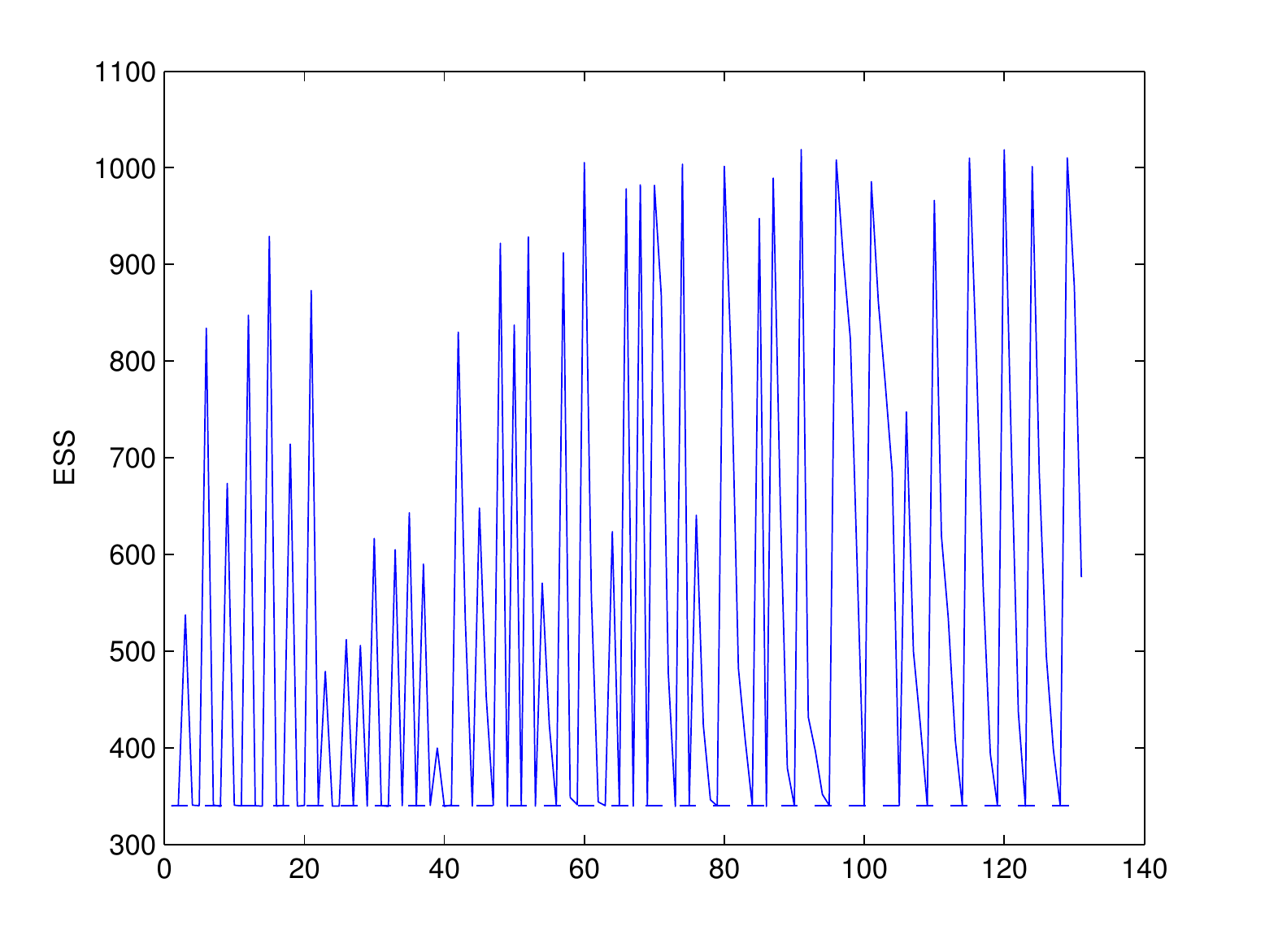}\includegraphics[width=0.45\textwidth,height=0.18\textheight]{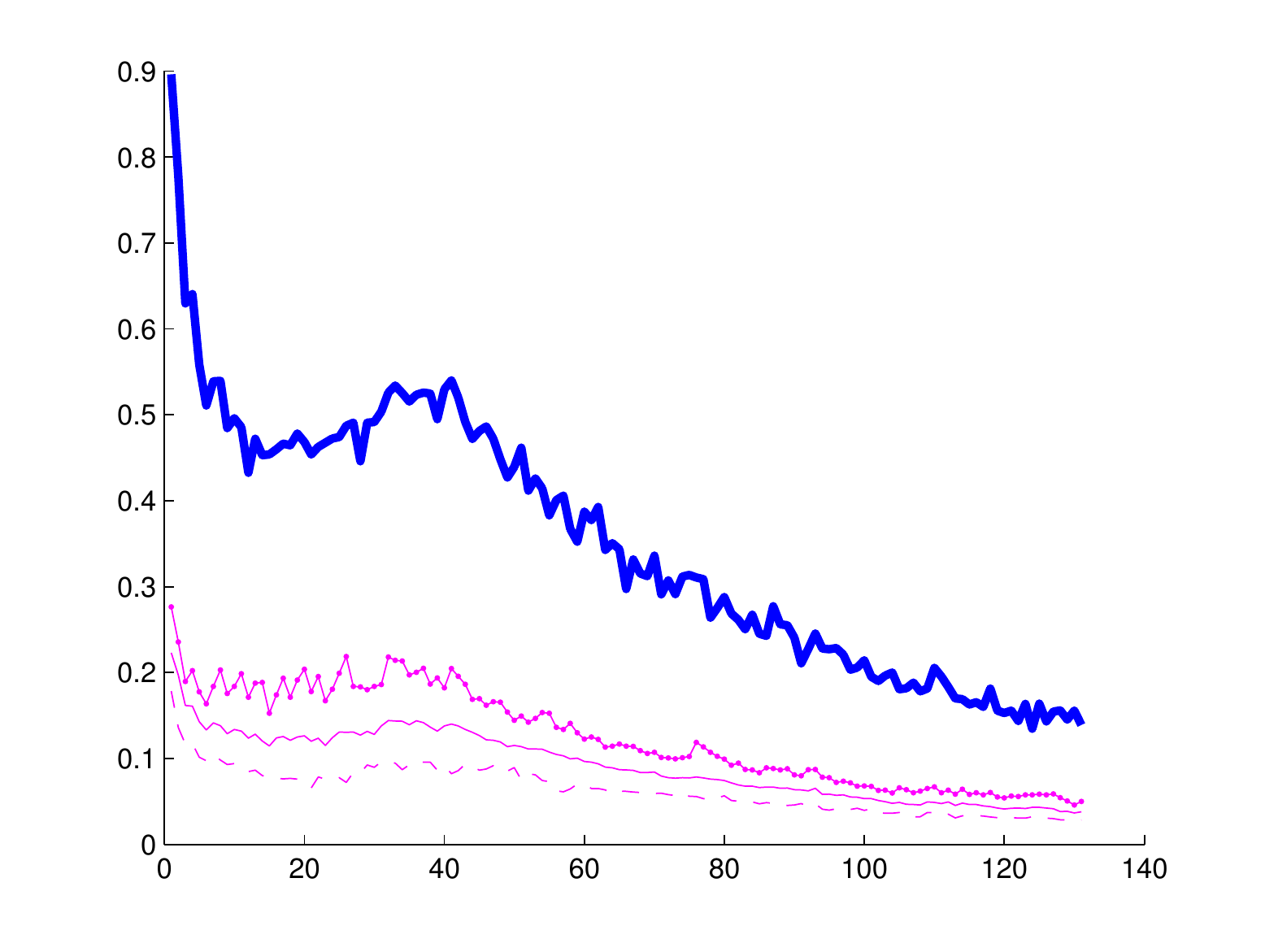}\\
 \includegraphics[width=0.45\textwidth,height=0.18\textheight]{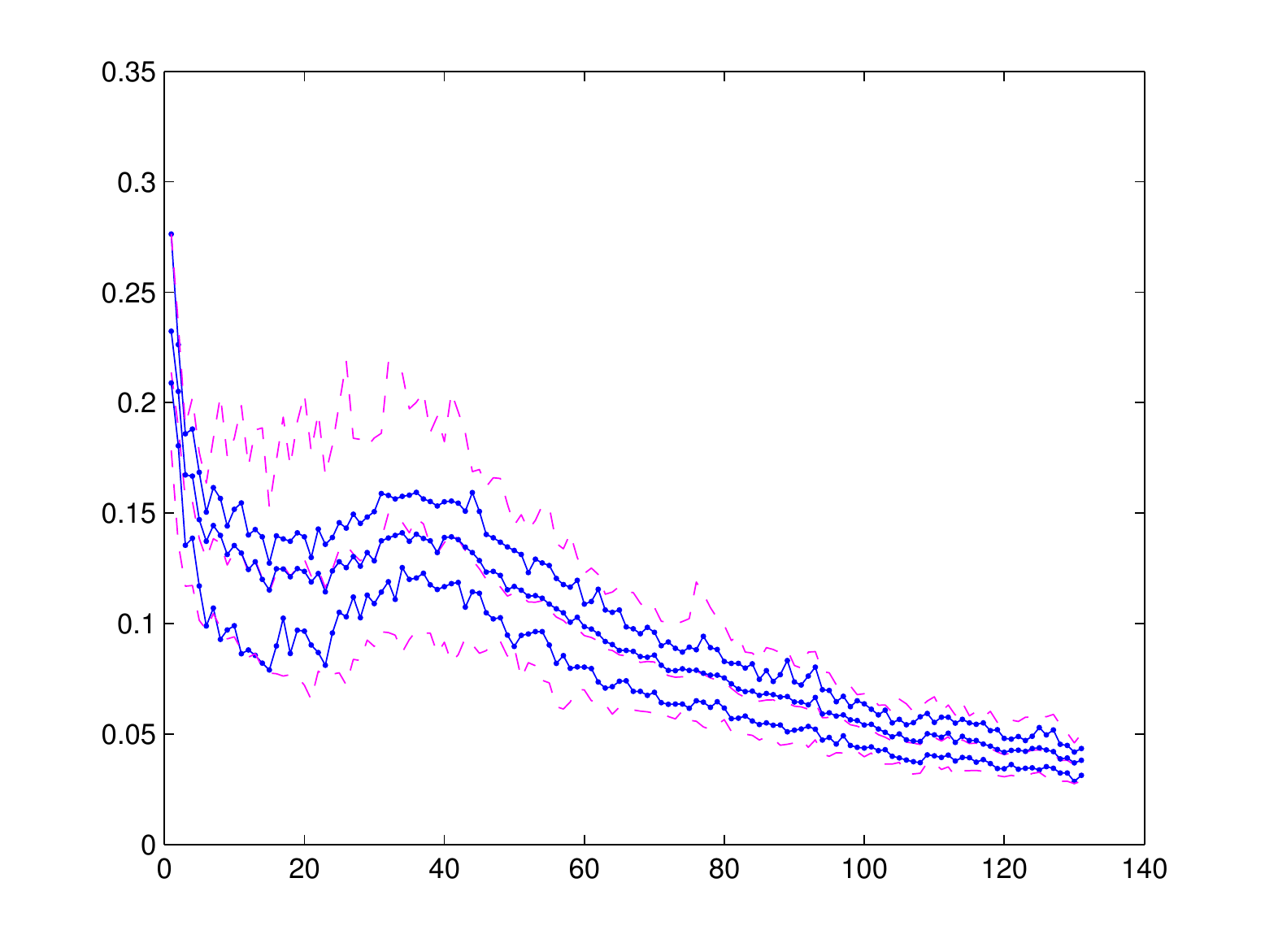}\includegraphics[width=0.45\textwidth,height=0.18\textheight]{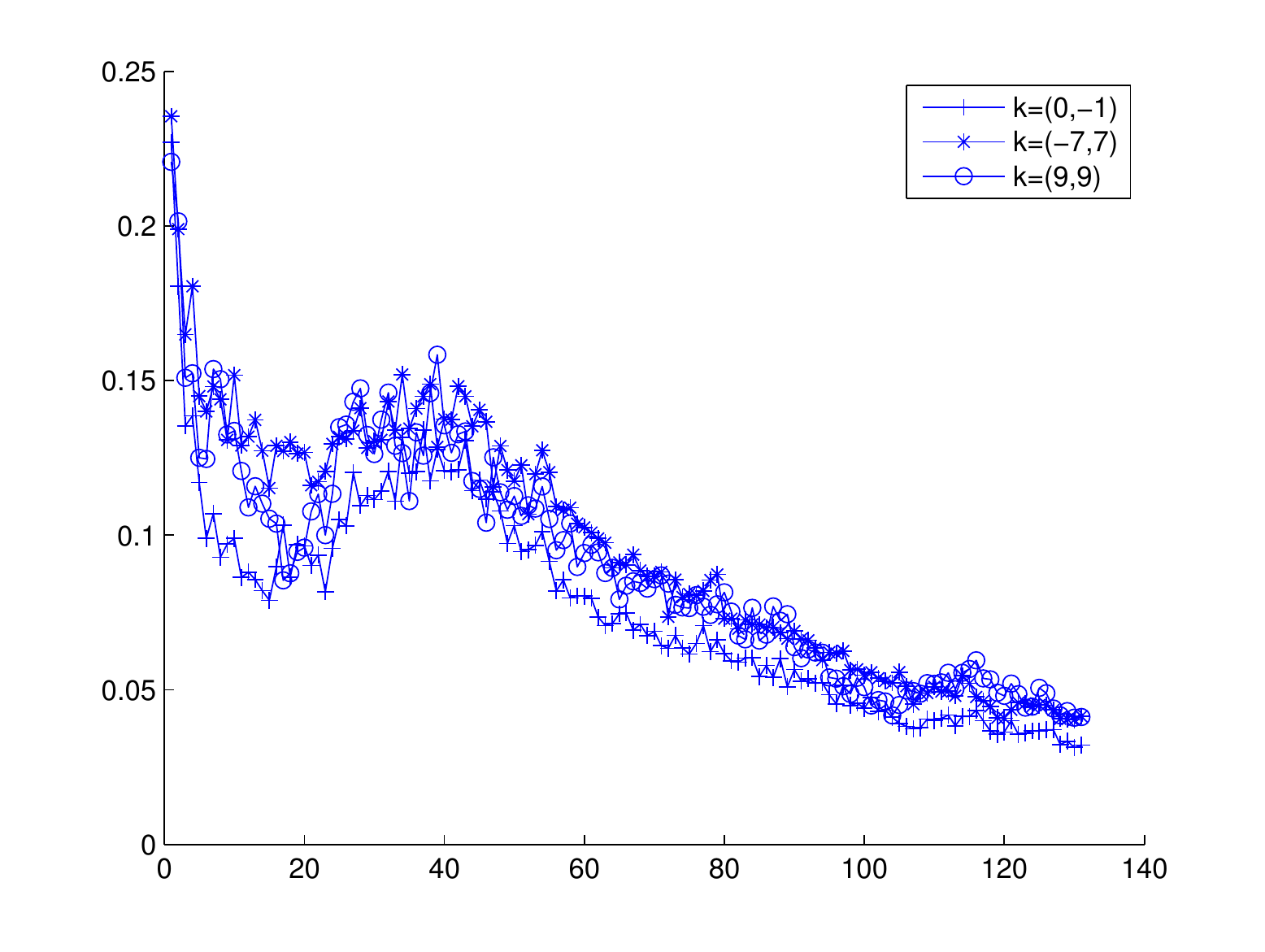}

\caption{Data-set A: monitoring SMC performance with $N=1020$. In all plots
the horizontal axis is index of SMC iteration $n,r$. Top left: ESS
oscillating between $N_{thresh}$ (when $\phi_{n,r}<1$) and higher
values (when $\phi_{n,q_{n}}=1$). Top right: thick-solid (blue) is
average acceptance ratio (w.r.t the particles), dot-solid (magenta)
is $\max_{k}J_{k,n,r}$, solid (magenta) is the average of $J_{k,n,r}$
(w.r.t $k$), dotted (magenta) is $\min_{k}J_{k,n,r}$. Bottom left:
We plot again maximum, minimum and average of $J_{k,n,r}$ w.r.t $k$
separately for $k\in\mathbf{K}\cap\mathbb{Z}_{\uparrow}^{2}$ (dash-dot,
blue) and $k\in\mathbf{K}^{c}\cap\mathbb{Z}_{\uparrow}^{2}$ (dashed,
magenta). Bottom right: $J_{k,n,r}$ for $k=(0,-1),(-7,7),(9,9)$.}

\label{ex1:monitorSMC} 
\end{figure}

\begin{figure}
\centering\includegraphics[width=1\textwidth,height=0.38\textheight]{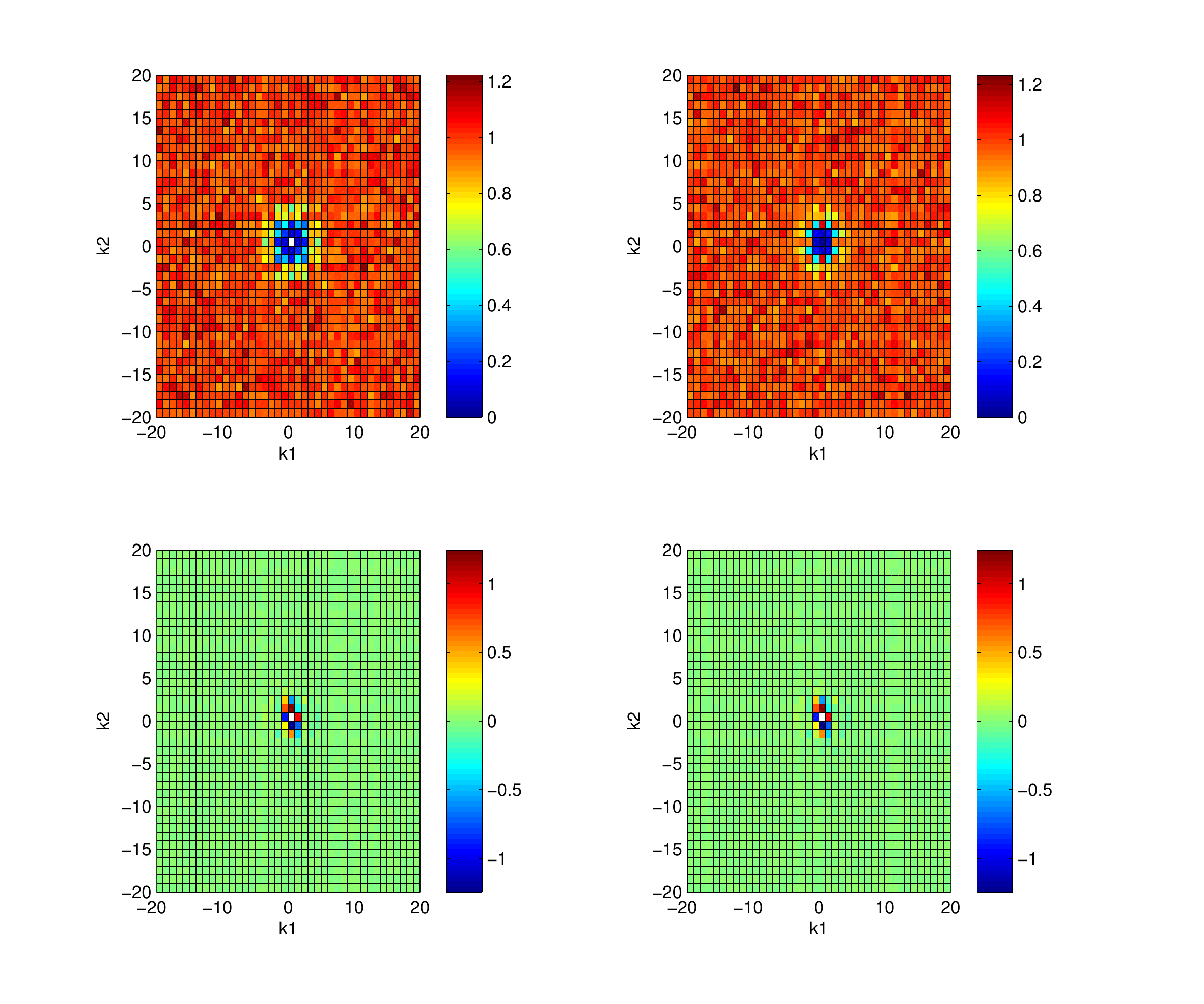}
\caption{Data-set A: posterior vs prior statistics for $\xi_{k,T}$ using SMC
with $N=1,020$ particles. Top: heat map of the ratio of estimated
posterior (marginal) standard deviations of $\xi_{k,T}$ over the
standard deviations of each $\xi_{k,0}$ (prior) against all frequencies
$k$. Bottom: corresponding heat map for the mean of each $\xi_{k,T}$
against $k$. The left and right plots correspond to the real and
imaginary parts respectively of the Fourier coefficients.}

\label{ex1:circle} 
\end{figure}

\begin{figure}
\centering\includegraphics[width=1\textwidth]{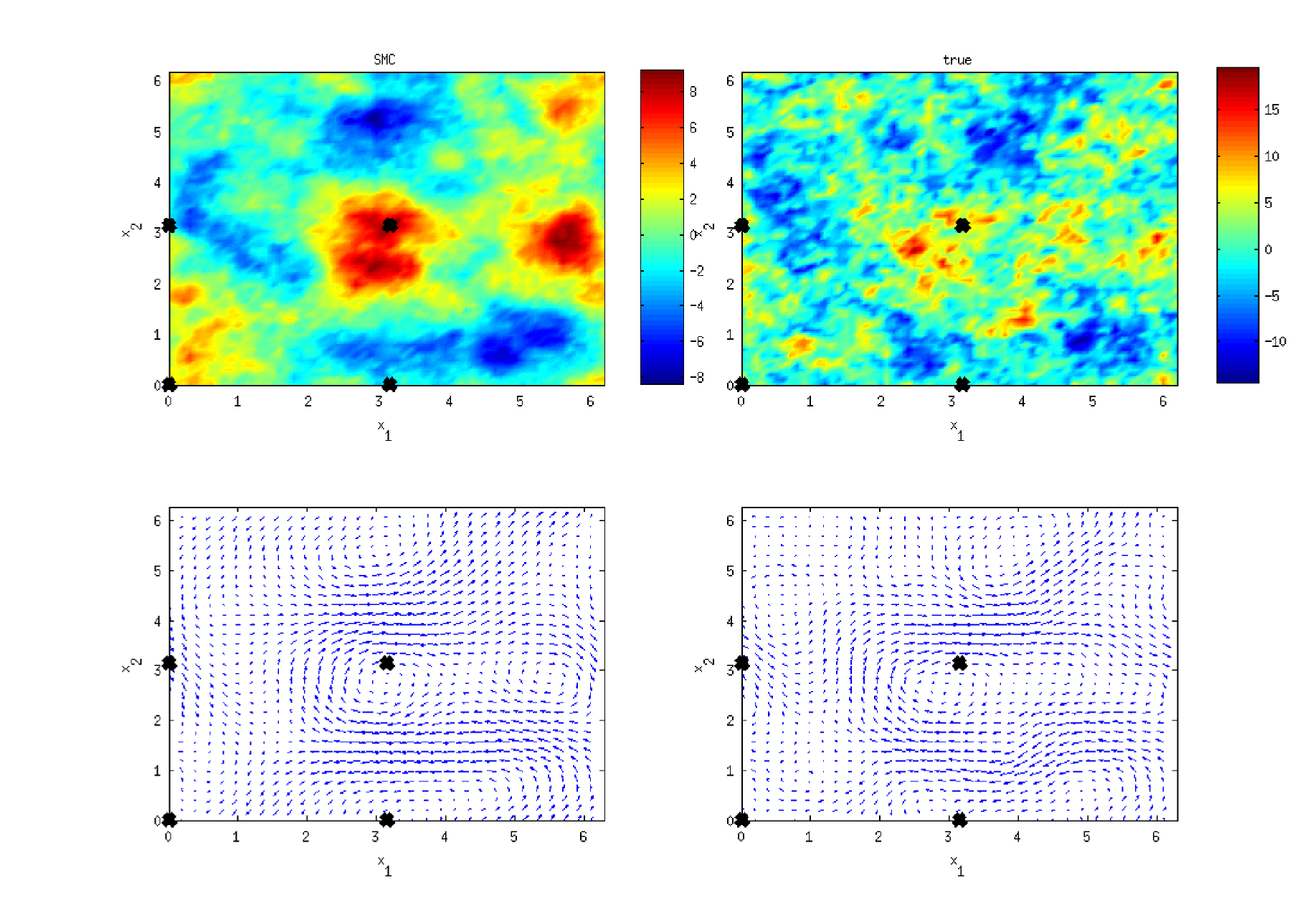}
\vspace{-0.8cm}
 \caption{Data-set B: Vorticity (top) and velocity field (bottom); posterior
mean (left) as estimated by SMC with $N=1020$ and true values (right).
The crosses indicate the positions $x_{1},\ldots,x_{\Upsilon}$ where
the vector field is observed. The graph is similar to Figure \ref{ex1:post_mean}
for Data-set A.}

\label{ex2:post_mean} 
\end{figure}

\begin{figure}
\centering\includegraphics[width=1\textwidth]{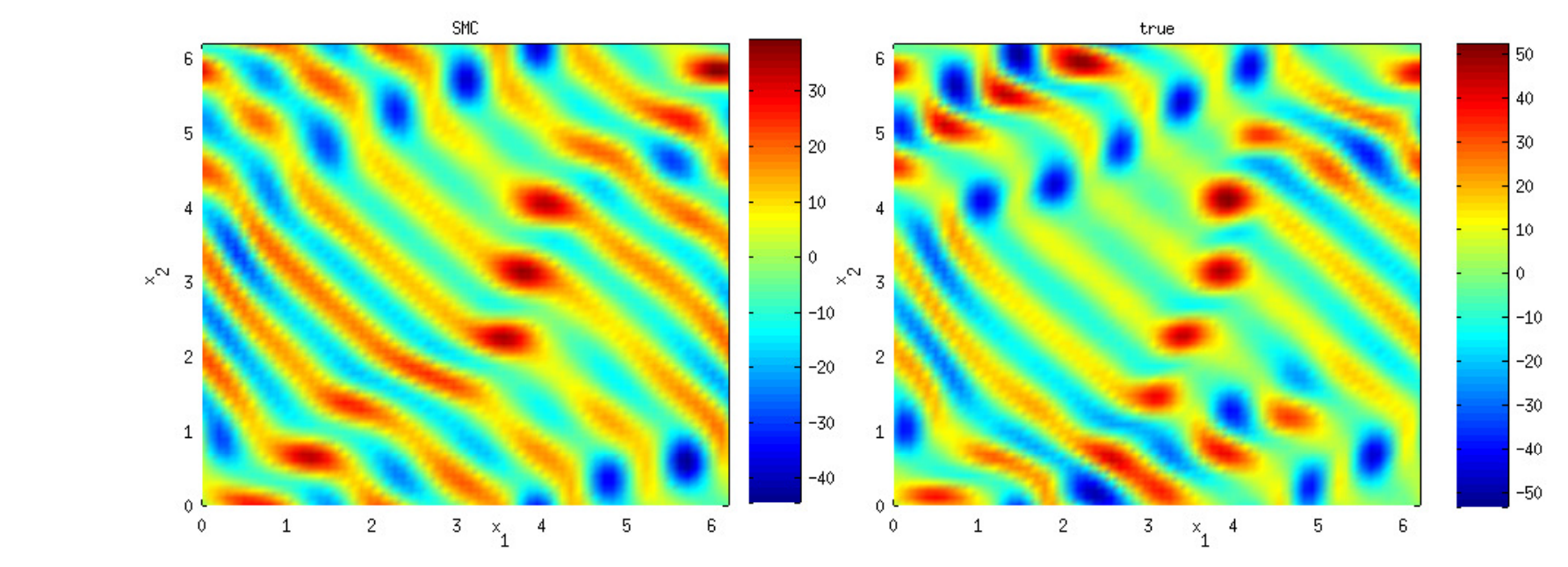}
\caption{Data-set B. Vorticity of $v(\cdot,\delta T)$. Posterior mean from
SMC (left) and true initial condition (right). }

\label{ex2:pred} 
\end{figure}

\begin{figure}
\centering\includegraphics[width=1\textwidth]{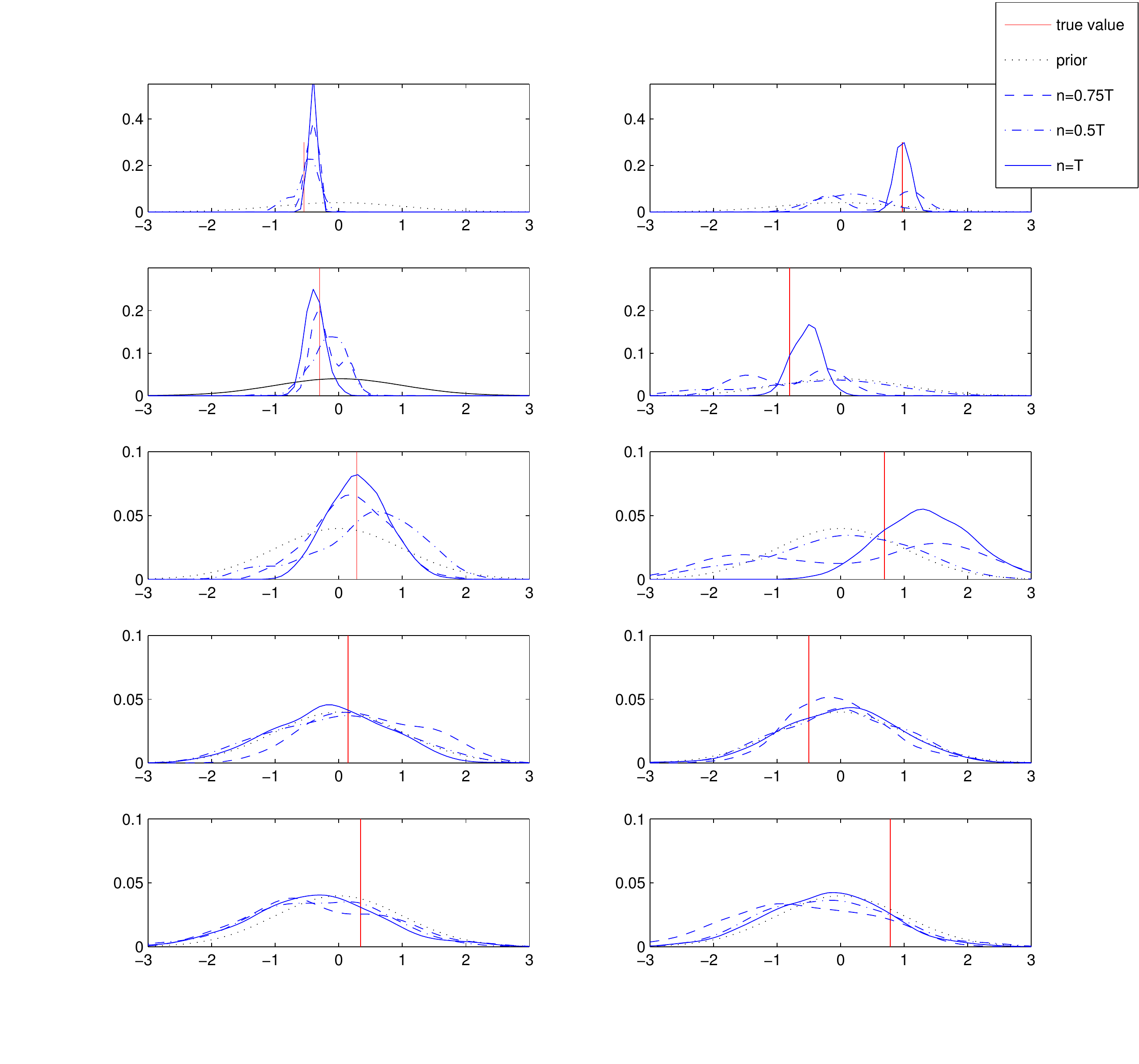} \caption{Data-set B: Estimated posterior PDFs for $\xi_{k,n}$ for $n=0,0.5T,0.75T,T$
and frequencies $k=(0,1),(1,1),(2,1),(4,4),(9,9)$. Details similarly
to Figure \ref{ex1:pdfs} for Data-set A.}

\label{ex2:pdfs} 
\end{figure}

\begin{figure}
\centering \includegraphics[width=19cm]{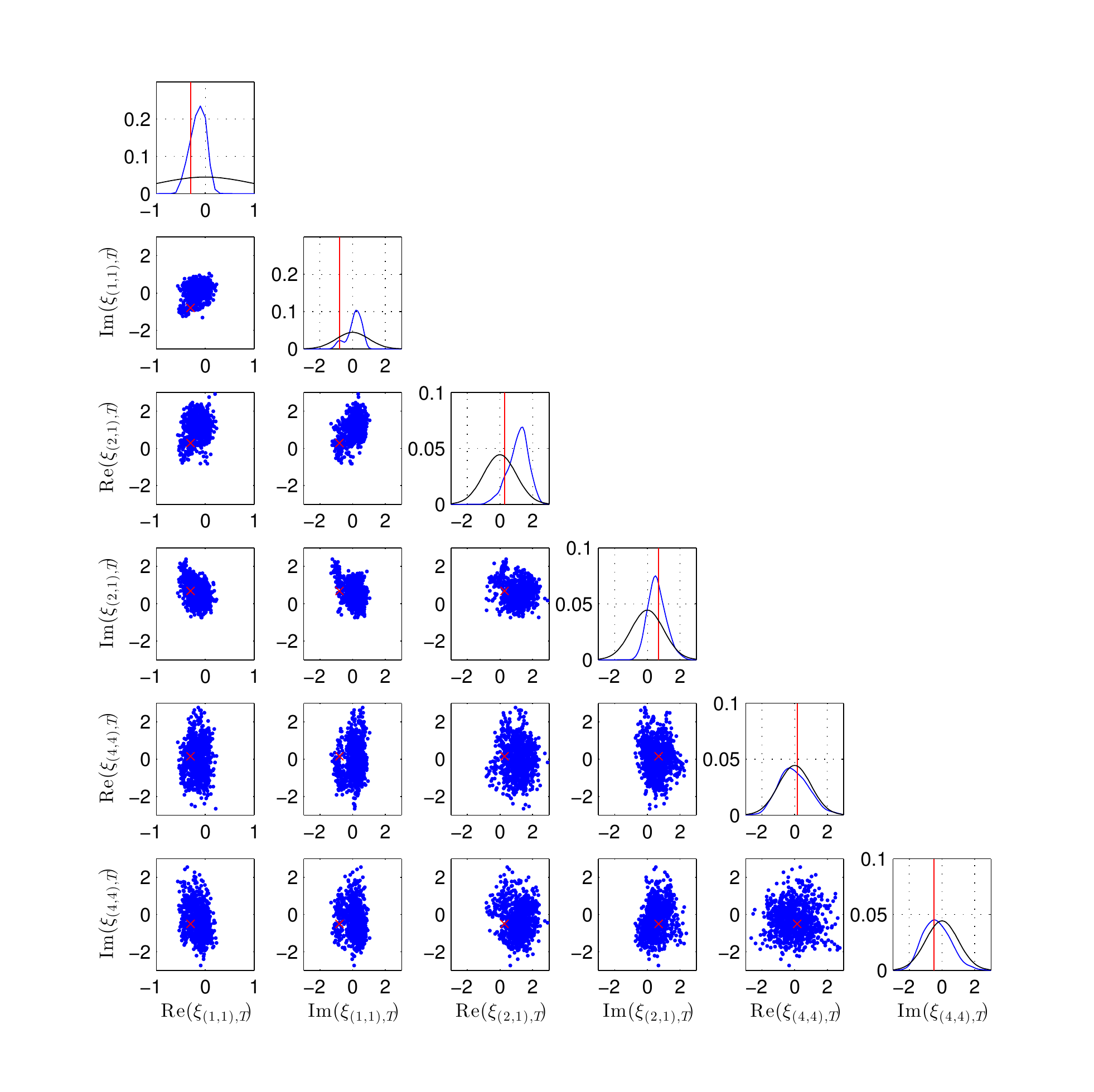} \caption{Data-set B: Scatter-plots generated from the SMC particles for frequencies
$k=(1,1),(2,1),(4,4)$. Details similarly to Figures \ref{ex1:scatterMCMC},
\ref{ex1:scatterSMC} for Data-set A.}

\label{ex2:scatterSMC} 
\end{figure}

\begin{figure}
\includegraphics[width=0.45\textwidth,height=0.18\textheight]{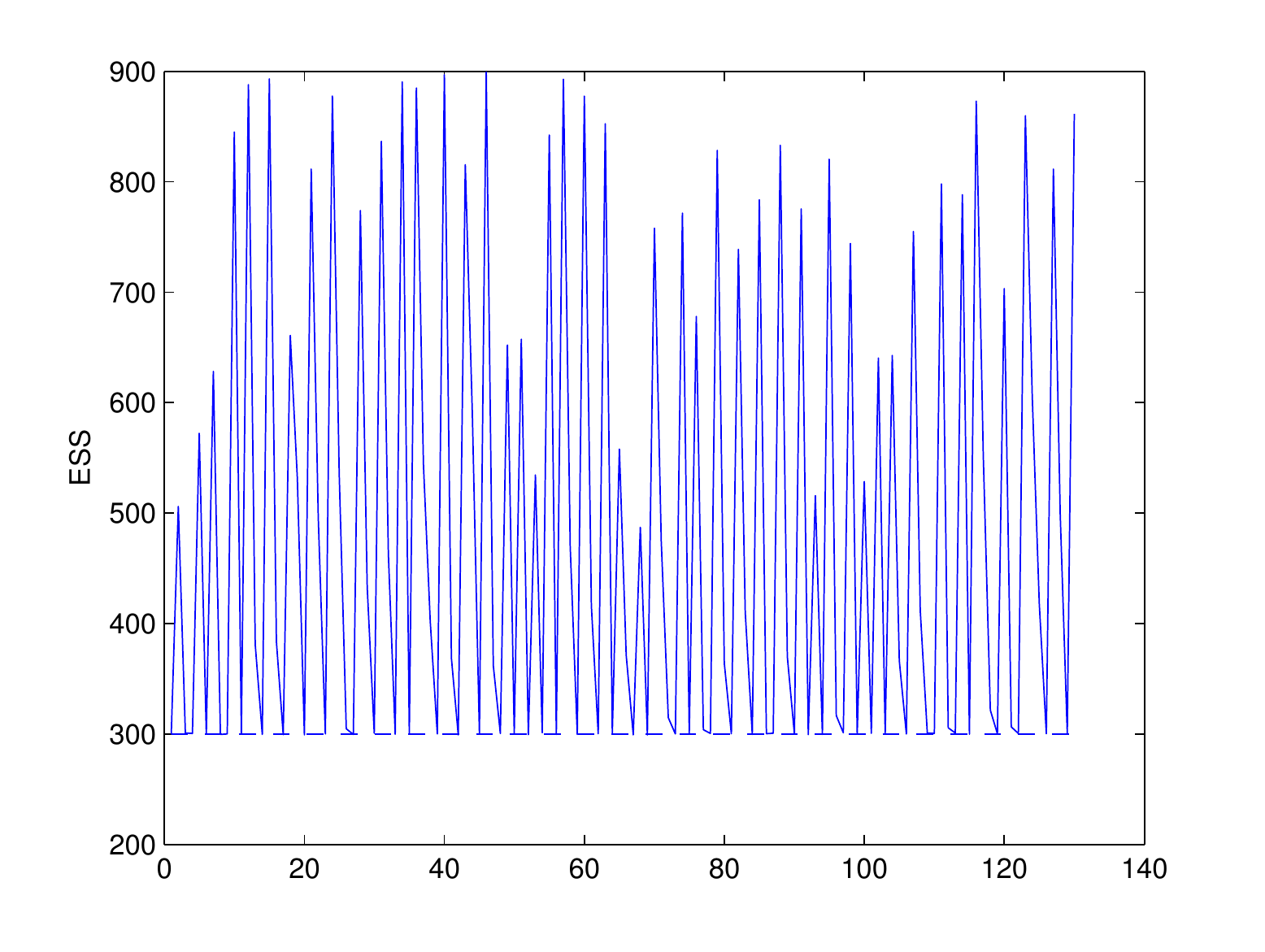}\includegraphics[width=0.45\textwidth,height=0.18\textheight]{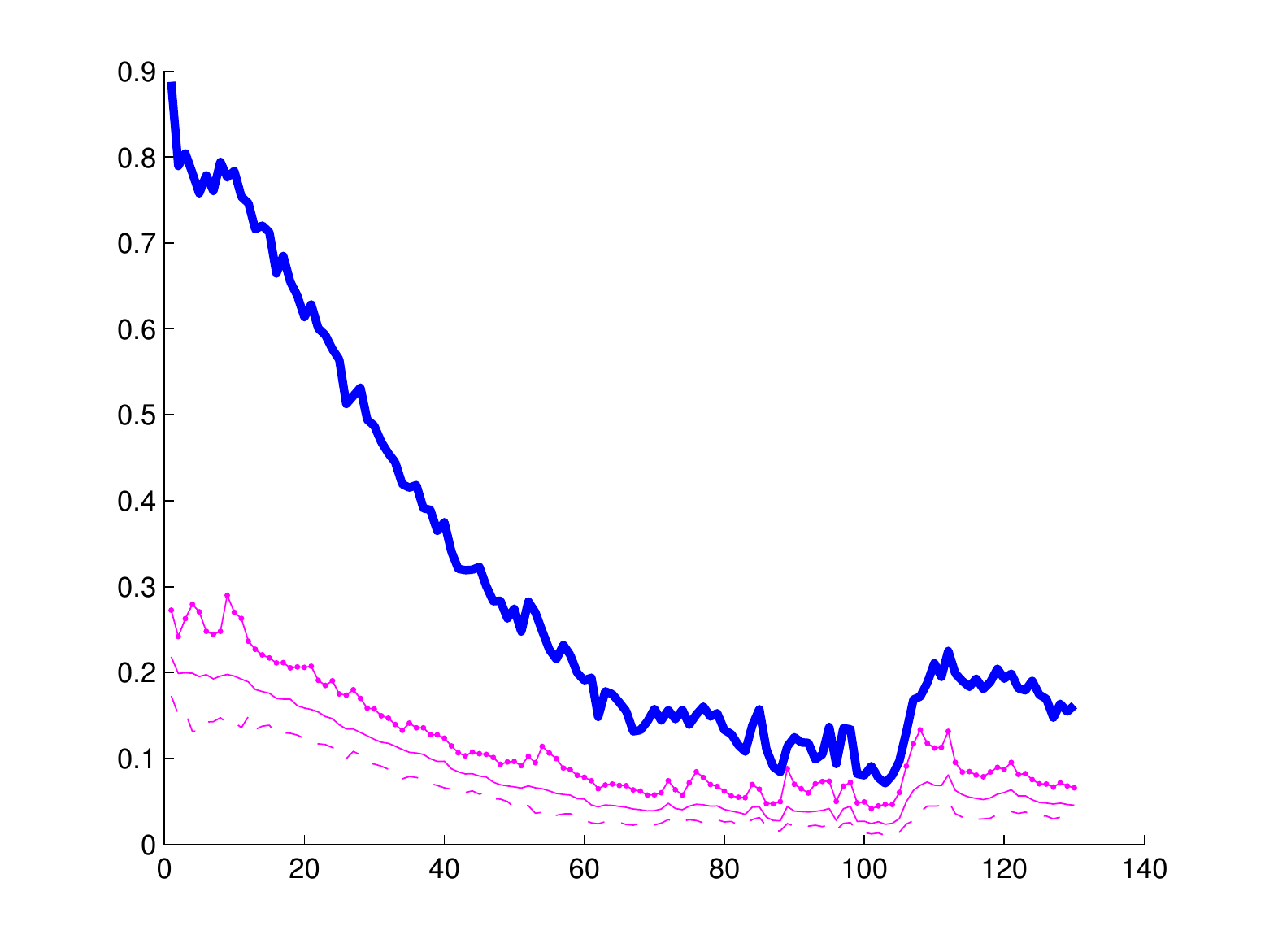}\\
 \includegraphics[width=0.45\textwidth,height=0.18\textheight]{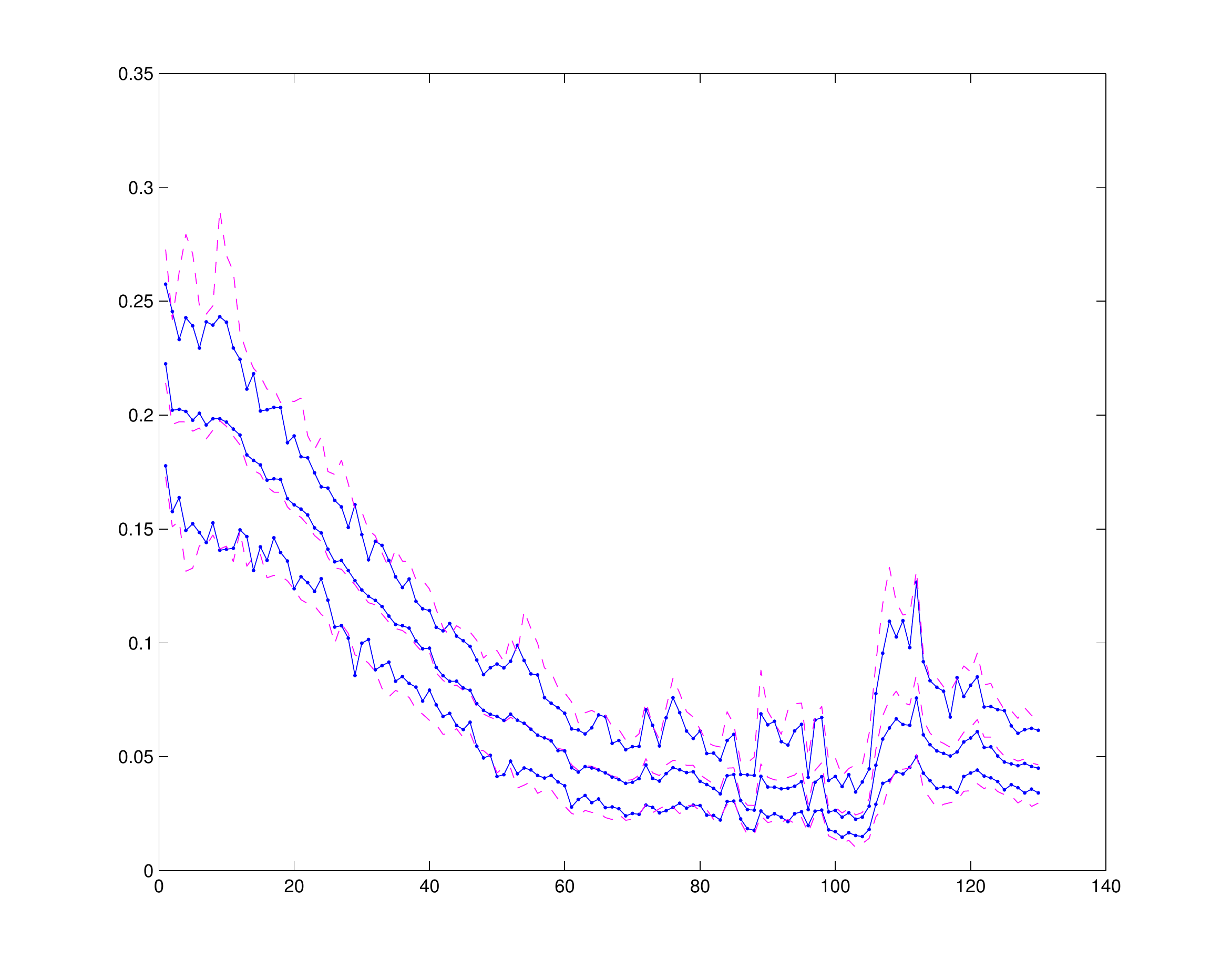}\includegraphics[width=0.45\textwidth,height=0.18\textheight]{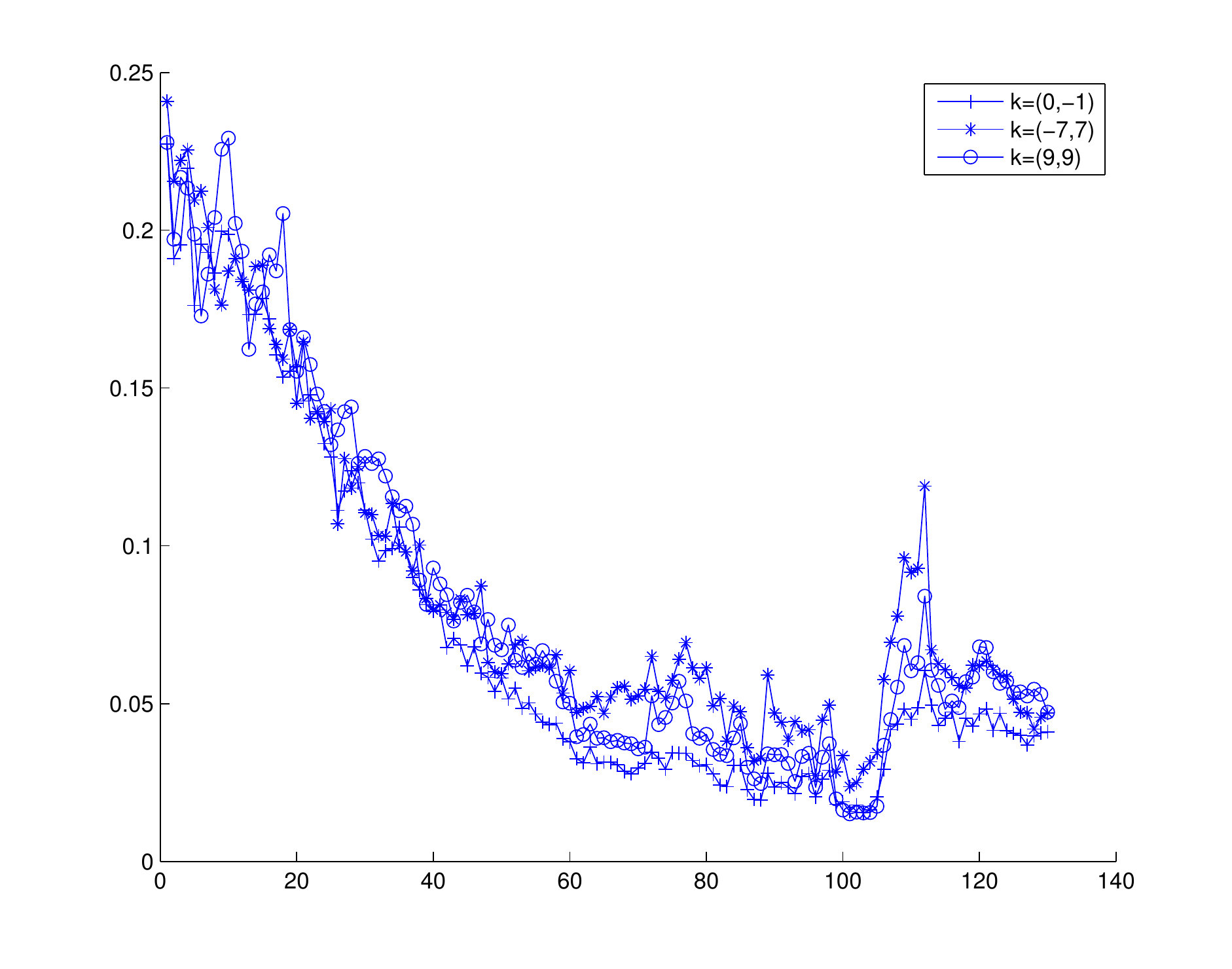}
\caption{Data-set B: monitoring SMC performance against iteration $n,r$. Top
left: ESS. Top right: average acceptance ratio (thick blue), $\max_{k}J_{k,n,r}$,
average of $J_{k,n,r}$, $\min_{k}J_{k,n,r}$ (thin-magenta). Bottom
left: $\max_{k}J_{k,n,r}$, average of $J_{k,n,r}$, $\min_{k}J_{k,n,r}$
considered separately for $k\in\mathbf{K}\cap\mathbb{Z}_{\uparrow}^{2}$
(blue) and $k\in\mathbf{K}^{c}\cap\mathbb{Z}_{\uparrow}^{2}$ (magenta).
Bottom right: $J_{k,n,r}$ for some values of $k$. For the full details
see caption in Figure \ref{ex1:monitorSMC} for Data-set A.}

\label{ex2:monitorSMC} 
\end{figure}

\begin{figure}
\centering\includegraphics[width=1\textwidth,height=0.38\textheight]{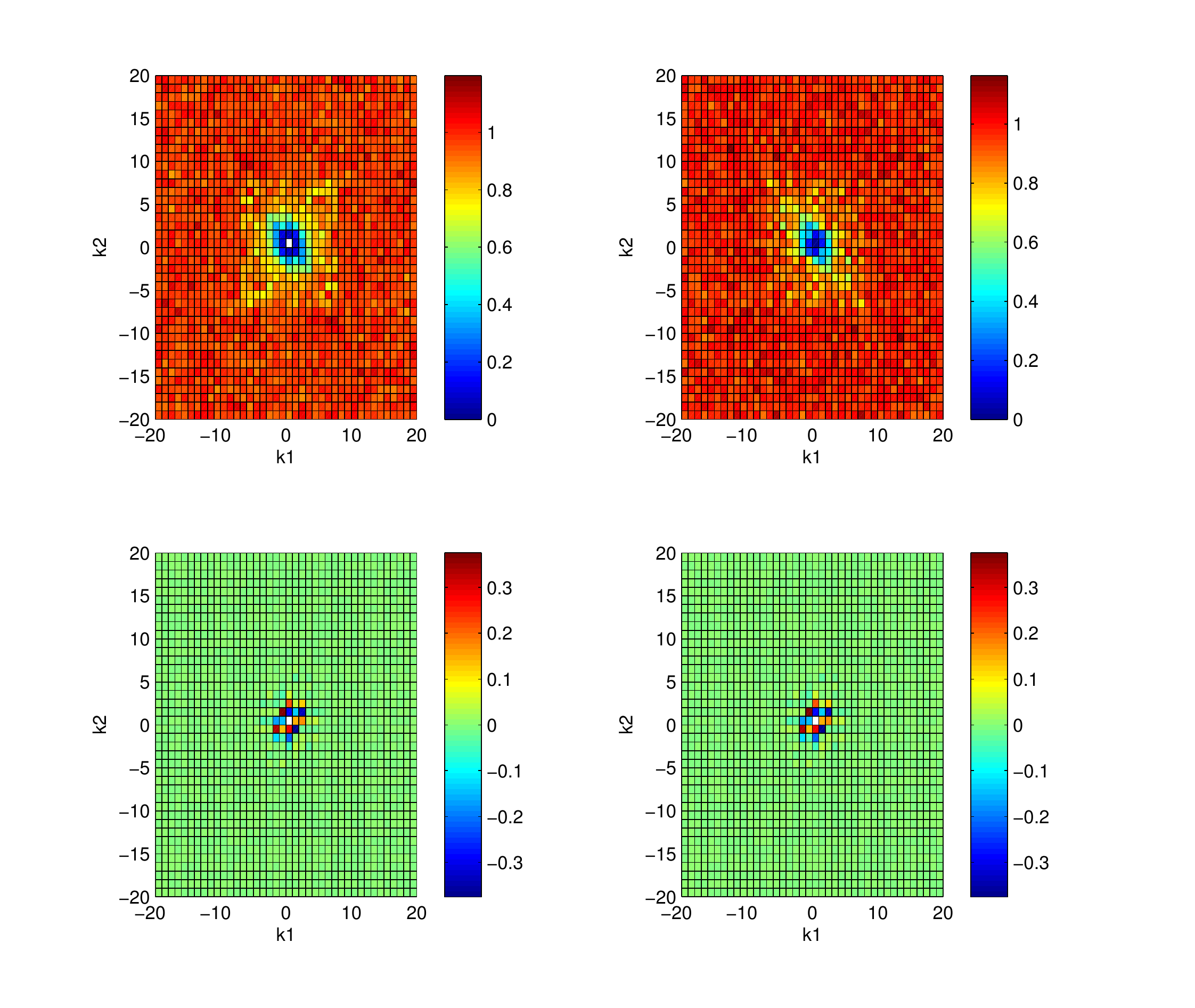}
\caption{Data-set B: heat map against $k$ of the ratio of the estimated (with
SMC) posterior standard deviations of $\xi_{k,T}$ to the one of the
prior (top) and posterior means for $\xi_{k,T}$. Details similarly
to Figure \ref{ex1:circle} for Data-set A.}

\label{ex2:circle} 
\end{figure}

\end{document}